\documentclass[11pt]{article}

\usepackage{amsthm}
\usepackage{graphicx} 
\usepackage{array} 

\usepackage{amsmath, amssymb, amsfonts, verbatim}
\usepackage{hyphenat,epsfig,multirow}

\usepackage[usenames,dvipsnames]{xcolor}
\usepackage[ruled]{algorithm2e}

\usepackage{tcolorbox}
\tcbuselibrary{skins,breakable}
\tcbset{enhanced jigsaw}

\usepackage[compact]{titlesec}

\definecolor{DarkRed}{rgb}{0.5,0.1,0.1}
\definecolor{DarkBlue}{rgb}{0.1,0.1,0.5}

\usepackage{hyperref}
\hypersetup{
colorlinks=true,
pdfnewwindow=true,
citecolor=ForestGreen,
linkcolor=DarkRed,
filecolor=DarkRed,
urlcolor=DarkBlue
}

\usepackage{tikz}
\usetikzlibrary{arrows}
\usetikzlibrary{arrows.meta}
\usetikzlibrary{shapes}
\usetikzlibrary{backgrounds}
\usetikzlibrary{positioning}
\usetikzlibrary{decorations.markings}
\usetikzlibrary{patterns}
\usetikzlibrary{calc}
\usetikzlibrary{fit}
\usetikzlibrary{snakes}

\usepackage{caption}
\usepackage{subcaption}

\usepackage{bm}
\usepackage{url}
\usepackage{xspace}
\usepackage[mathscr]{euscript}

\usepackage{mdframed}

\usepackage[noend]{algpseudocode}
\makeatletter
\def\BState{\State\hskip-\ALG@thistlm}
\makeatother

\usepackage{cite}
\usepackage{enumitem}

\usepackage[margin=1in]{geometry}

\newtheorem{theorem}{Theorem}
\newtheorem{lemma}{Lemma}[section]
\newtheorem{proposition}[lemma]{Proposition}

\newtheorem{claim}[lemma]{Claim}

\newtheorem{definition}{Definition}

\newtheorem{problem}{Problem}
\newtheorem{remark}[lemma]{Remark}

\newtheorem*{claim*}{Claim}
\newtheorem*{proposition*}{Proposition}
\newtheorem*{lemma*}{Lemma}
\newtheorem*{problem*}{Problem}

\newtheorem{mdresult}[theorem]{Theorem}
\newenvironment{Theorem}{\begin{mdframed}[backgroundcolor=lightgray!40,topline=false,rightline=false,leftline=false,bottomline=false,innertopmargin=2pt]\begin{mdresult}}{\end{mdresult}\end{mdframed}}

\newtheorem{mdinvariant}{Invariant}

\allowdisplaybreaks

\renewcommand{\qed}{\nobreak \ifvmode \relax \else
      \ifdim\lastskip<1.5em \hskip-\lastskip
      \hskip1.5em plus0em minus0.5em \fi \nobreak
      \vrule height0.75em width0.5em depth0.25em\fi}

\newcommand{\Eq}[1]{\ensuremath{\underset{\textnormal{#1}}=}}

\newcommand{\sH}{\ensuremath{\# H}}

\newcommand{\Os}{\ensuremath{O^*}}

\newcommand{\toShrink}{-.20cm}
\newcommand{\toShrinkEnu}{-.2cm}

\newcommand{\Ot}{\ensuremath{\widetilde{O}}}
\newcommand{\eps}{\ensuremath{\varepsilon}}
\newcommand{\Paren}[1]{\Big(#1\Big)}
\newcommand{\Bracket}[1]{\Big[#1\Big]}
\newcommand{\bracket}[1]{\left[#1\right]}
\newcommand{\paren}[1]{\ensuremath{\left(#1\right)}\xspace}
\newcommand{\card}[1]{\left\vert{#1}\right\vert}

\newcommand{\IN}{\ensuremath{\mathbb{N}}}

\newcommand{\ceil}[1]{{\left\lceil{#1}\right\rceil}}

\newcommand{\expect}[1]{\Exp\bracket{#1}}
\newcommand{\var}[1]{\ensuremath{\mathbb{V}\textnormal{ar}}\bracket{#1}}

\newcommand{\set}[1]{\ensuremath{\left\{ #1 \right\}}}
\newcommand{\poly}{\mbox{\rm poly}}
\newcommand{\polylog}{\mbox{\rm  polylog}}

\newcommand{\alg}{\ensuremath{\mathcal{A}}\xspace}

\DeclareMathOperator*{\Var}{\ensuremath{\mathbb{V}\textnormal{ar}}}
\DeclareMathOperator*{\Exp}{\ensuremath{{\mathbb{E}}}}
\DeclareMathOperator*{\Prob}{\ensuremath{\textnormal{Pr}}}
\renewcommand{\Pr}{\Prob}

\newcommand{\Ex}{\Exp}

\newcommand{\etal}{{\it et al.\,}}

\newcommand{\FG}{\ensuremath{\mathcal{G}}}

\newcommand{\event}[1]{\ensuremath{{\sf E}_{#1}}}

\newenvironment{tbox}{\begin{tcolorbox}[
		enlarge top by=5pt,
		enlarge bottom by=5pt,
		 drop shadow=black,
		 boxsep=0pt,
                  left=4pt,
                  right=4pt,
                  top=10pt,
                  arc=0pt,
                  boxrule=1pt,toprule=1pt,
                  colback=white
                  ]
	}
{\end{tcolorbox}}

\newcommand{\Prot}{\ensuremath{\Pi}}

\newcommand{\bR}{\bm{R}}

\newcommand{\supp}[1]{\ensuremath{\textsc{supp}(#1)}}

\renewcommand{\event}{\mathcal{E}}

\newcommand{\II}{\ensuremath{\mathbb{I}}}

\newcommand{\CC}{\ensuremath{\mathcal{C}}}

\newcommand{\DD}{\ensuremath{\mathcal{D}}}

\renewcommand{\SS}{\ensuremath{\mathcal{S}}}

\newcommand{\rhoC}{\rho^C}

\newcommand{\rhoS}{\rho^S}

 \newcommand{\rhoG}[1]{\rho_{#1+}}

\newcommand{\scestimator}{\ensuremath{\textnormal{\texttt{odd-cycle-sampler}}}\xspace}
\newcommand{\ssestimator}{\ensuremath{\textnormal{\texttt{star-sampler}}}\xspace}
\newcommand{\ssampler}{\ensuremath{\textnormal{\texttt{subgraph-sampler}}}\xspace}

\newcommand{\ustar}{\ensuremath{u^{*}}}
\newcommand{\dstar}{\ensuremath{d^{*}}}

\newcommand{\sC}{\ensuremath{\# C_{2k+1}}}
\newcommand{\ssC}[1]{\ensuremath{(\sC\mid{#1})}}

\newcommand{\sS}{\ensuremath{\# S_{\ell}}}

\renewcommand{\alg}{\textsc{alg}}

\newcommand{\be}{\bm{e}}
\newcommand{\bw}{\bm{w}}

\newcommand{\bL}[1]{\ensuremath{\bm{L}_{#1}}}

 \newcommand{\TT}{\ensuremath{\mathcal{T}}}
 \newcommand{\PP}{\ensuremath{\mathcal{P}}}
 \newcommand{\LL}{\ensuremath{\mathcal{L}}}
 
 \newcommand{\Label}[1]{\ensuremath{\textnormal{\textsf{label}}[#1]}\xspace}
 \newcommand{\Value}[1]{\ensuremath{\textnormal{\textsf{value}}[#1]}\xspace}

 \newcommand{\ssH}[1]{\ensuremath{(\sH\mid{#1})}}
 
 \newcommand{\tH}{\ensuremath{\widetilde{H}}}

\newcommand{\ve}{\vec{e}}
 \newcommand{\vE}{\vec{E}}

\title{A Simple Sublinear-Time Algorithm for Counting Arbitrary Subgraphs via Edge Sampling}
\author{Sepehr Assadi\thanks{Department of Computer and Information Science, University of Pennsylvania. Supported in part by the National Science Foundation grant CCF-1617851. Email: \texttt{sassadi@cis.upenn.edu}.}   \and 
\and Michael Kapralov\thanks{School of Computer and Communication Sciences, EPFL. Email: \texttt{michael.kapralov@epfl.ch}.}   \and 
\and Sanjeev Khanna\thanks{Department of Computer and Information Science, University of Pennsylvania. Supported in part by the National Science Foundation grants CCF-1617851 and CCF-1763514. Email: \texttt{sanjeev@cis.upenn.edu}.} 
}

\date{}

\begin{document}
\maketitle

\thispagestyle{empty}
\begin{abstract}
In the subgraph counting problem, we are given a (large) input graph $G(V, E)$ and a (small) target graph $H$ (e.g., a triangle); the goal is to estimate the number of occurrences of $H$ in $G$. Our 
focus here is on designing \emph{sublinear-time} algorithms for approximately computing number of occurrences of $H$ in $G$ in the setting where the algorithm is given \emph{query access} to $G$.
This problem has been studied in several recent papers which primarily focused on specific families of graphs $H$ such as triangles, cliques, and stars. However, not 
much is known about approximate counting of arbitrary graphs $H$ in the literature. This is in sharp contrast to the closely related subgraph enumeration problem that has received significant attention in the database community as the database join problem.
The AGM bound shows that the maximum number of occurrences of any arbitrary subgraph $H$ in a graph $G$ with $m$ edges is $O(m^{\rho(H)})$, where $\rho(H)$ is the \emph{fractional edge-cover} of $H$, 
and enumeration algorithms with matching runtime are known for any $H$. 

\medskip

In this work, we bridge this gap between the subgraph counting and subgraph enumeration problems by designing a simple sublinear-time algorithm that can estimate the number of occurrences of 
any arbitrary graph $H$ in $G$, denoted by $\sH$, to within a $(1\pm \eps)$-approximation with high probability in $O(\frac{m^{\rho(H)}}{\sH}) \cdot \poly(\log{n},1/\eps)$ time. Our algorithm is allowed 
the standard set of queries for general graphs, namely degree queries, pair queries and neighbor queries, plus an additional edge-sample query that returns an edge chosen uniformly at random.
The performance of our algorithm matches those of Eden~\etal [FOCS 2015, STOC 2018] for counting triangles and cliques and extend them to all choices of subgraph $H$ under 
the additional assumption of edge-sample queries. We further show that our algorithm works for a more general version of the estimation problem where edges of $G$ and $H$ have colors, which corresponds to the database join size estimation problem. 
For this slightly more general setting, we also establish a matching lower bound for any choice of subgraph $H$.

\end{abstract}
\clearpage
\setcounter{page}{1}


\section{Introduction}\label{sec:intro}

Counting (small) subgraphs in massive graphs is a fundamental algorithmic problem, with a wide range of applications in bioinformatics, social network analysis, spam detection and graph databases (see, e.g.~\cite{milo, irs, sociology-1}).  In social network analysis, the ratio of the number of triangles in a network to the number of length $2$ paths  (known as the {clustering coefficient}) as well as subgraph counts for larger subgraphs $H$ have been proposed as important metrics for analyzing massive networks~\cite{Ugander}. Similarly, {motif counting} are popular method for analyzing protein-protein interaction networks in bioinformatics (e.g.,~\cite{milo}). In this paper we consider designing efficient algorithms for this task.

Formally, we consider the following problem: Given a (large) graph $G(V, E)$ with $m$ edges and a (small) subgraph $H(V_H,E_H)$ (e.g., a triangle) and a precision parameter $\eps \in (0, 1)$,  output a $(1\pm\eps)$-approximation\footnote{Throughout, we say that $a$ is a $(1\pm \eps)$ approximation to $b$ iff $(1-\eps) \cdot b \leq a \leq (1+\eps) \cdot b$. } to the number of occurrences of $H$ in $G$. Our goal is to design an algorithm that runs in time {\em sublinear} in the number $m$ of edges of $G$, and in particular makes a sublinear number of the following types of queries to the graph $G$: 
\vspace{-5pt}
\begin{itemize}
	\item {\bf Degree query $v$:} the degree $d_v$ of any vertex $v \in V$; 
	\item {\bf Neighbor query $(v,i)$:} what vertex is the $i$-th neighbor of the vertex $v \in V$ for $i \leq d_v$; 
	\item {\bf Pair query $(u,v)$:} test for pair of vertices $u,v \in V$, whether or not $(u,v)$ is an edge in $E$. 
	\item {\bf Edge-sample query:} sample an edge $e$ uniformly at random from $E$.
\end{itemize}
\vspace{-5pt}
The first three queries are the standard baseline queries (see Chapter 10 of Goldreich's book~\cite{Goldreich17}) 
 assumed by nearly all sublinear time algorithms for counting small subgraphs such as triangles or cliques~\cite{EdenLRS15,EdenRS18} 
(see~\cite{GonenRS10} for the necessity of using pair queries for counting subgraphs beside stars). 
The last query is somewhat less standard but is also considered in the literature prior to our work, for example in~\cite{AliakbarpourBGP18} for counting stars in sublinear time, and in~\cite{EdenR18b} in the context 
of lower bounds for subgraph counting problems.

\subsection{Our Contributions}\label{sec:results}

For the sake of clarity, we suppress any dependencies on the approximation
 parameter $\eps$, $\log{n}$-terms, and the size of graph $H$, using the $\Os(\cdot)$ notation. 
 Our results are parameterized by the \emph{fractional edge-cover number} of the subgraph $H$ (see Section~\ref{sec:edge-cover} for the formal definition). Our goal in this paper is to approximately compute the number of occurrences $\sH$ of $H$ in $G$, 
 formally defined as number of subgraphs $H'$ of $G$ such that $H$ and $H'$ are isomorphic to each other. 

\begin{Theorem}\label{thm:main-intro}
	There exists a randomized algorithm that given a precision parameter $\eps \in (0,1)$, a subgraph $H$, and a query access to the input graph $G$, 
	with high probability outputs a $(1\pm \eps)$ approximation to the number of occurrences of $H$ in $G$, denoted by $\sH$, using: 
	\begin{align*}
		\Os\Paren{\min\set{m,\frac{m^{\rho(H)}}{\sH}}} \textnormal{ queries and } \Os\Paren{\frac{m^{\rho(H)}}{\sH}} \textnormal{ time.}
	\end{align*}
	The algorithm uses degree queries, neighbor queries, pair queries, and edge-sample queries. 
\end{Theorem}

Since the fractional edge-cover number of any $k$-clique $K_k$ is $k/2$, as a corollary of Theorem~\ref{thm:main-intro}, we obtain sublinear algorithms for counting triangles, and 
in general $k$-cliques using 
\begin{align*}
\Os\paren{\min\set{m,\frac{m\sqrt{m}}{\#K_3}}} \textnormal{ and } \Os\paren{\min\set{m,\frac{m^{k/2}}{\#K_k}}},
\end{align*}
queries respectively. These bounds match the previous results of Eden~\etal~\cite{EdenLRS15,EdenRS18} modulo an additive term of $\Os(\frac{n}{(\#K_3)^{1/3}})$ for triangles in~\cite{EdenLRS15} and $\Os(\frac{n}{(\#K_k)^{1/k}})$ for arbitrary cliques in~\cite{EdenRS18} which is needed in the absence of edge-sample queries that are not used by~\cite{EdenLRS15,EdenRS18}. 
Our bounds settle a conjecture of Eden and Rosenbaum~\cite{EdenR18b}  in the affirmative by showing that one can avoid 
the aforementioned additive terms depending on $n$ in query complexity by allowing edge-sample queries. We now elaborate more on different aspects of our result in Theorem~\ref{thm:main-intro}.

\paragraph{AGM Bound and Database Joins.} The problem of {\em enumerating} all occurrences of a graph $H$ in a graph $G$ and, more generally, the database join problem, has been considered extensively in the literature. A fundamental 
question here is that given a database with $m$ entries (e.g. a graph $G$ with $m$ edges) and a conjunctive query $H$ (e.g. a small graph $H$), what is the maximum possible size of the output of the query (e.g., number of occurrences of $H$ in $G$)? 
The AGM bound of Atserias, Grohe and Marx~\cite{AtseriasGM08} provides a tight upper bound of $m^{\rho(H)}$ (up to constant factors), where $\rho(H)$ is the fractional edge cover of $H$. The AGM bound applies to databases with 
several relations, and the fractional edge cover in question is weighted according to the sizes of the different relations. A similar bound on the number of occurrences of a hypergraph $H$ inside a hypergraph $G$ with $m$ hyperedges was proved 
earlier by Friedgut and Kahn~\cite{FriedgutKahn}, and the bound for graphs is due to Alon~\cite{Alon}. Recent work of Ngo et al.~\cite{NgoPRR18} gave algorithms for evaluating database joins in time bounded by worst case output size for a database 
with the same number of entries. When instantied for the subgraph enumeration problem, their result gives an algorithm for enumerating all occurrences of $H$ in a graph $G$ with $m$ edges in time $O(m^{\rho(H)})$. 

Our Theorem~\ref{thm:main-intro} is directly motivated by the connection between subgraph counting and subgraph enumeration problems and the AGM bound. 
In particular, Theorem~\ref{thm:main-intro} provides a natural analogue of AGM bound for estimation algorithms by stating that 
if the number of occurrences $H$ is $\sH \leq m^{\rho(H)}$, then a $(1\pm\eps)$-approximation to $\sH$ can be obtained in $\Os(\frac{m^{\rho(H)}}{\sH})$ time.
Additionally, as we show in Section~\ref{sec:extensions}, Theorem~\ref{thm:main-intro} can be easily extended to the more general problem of database join size estimation (for binary relations). This problem corresponds to a subgraph counting problem in which 
the graphs $G$ and $H$ are both \emph{edge-colored} and our goal is to count the number of copies of $H$ in $G$ with the same colors on edges. Our algorithm 
can solve this problem also in $\Os(\frac{m^{\rho(H)}}{\sH_c})$ time where $\sH_c$ denotes the number of copies of $H$ with the same colors in $G$. 


\paragraph{Optimality of Our Bounds.} Our algorithm in Theorem~\ref{thm:main-intro} is optimal from different points of view. Firstly, by a lower bound of~\cite{EdenR18b} 
(building on~\cite{EdenLRS15,EdenRS18}), the bounds achieved by our algorithm when $H$ is any $k$-clique is optimal among all algorithms with the same query access (including the edge-sample query). 
In Theorem~\ref{thm:lb-odd}, we further prove a lower bound showing that for \emph{odd cycles} as well, the bounds achieved by Theorem~\ref{thm:main-intro} are optimal. These results hence suggest that Theorem~\ref{thm:main-intro} is \emph{existentially 
optimal}: there exists several natural choices for $H$ such that Theorem~\ref{thm:main-intro} achieves the optimal bounds. However, there also exist choices of $H$ for which the bounds in Theorem~\ref{thm:main-intro} are suboptimal. 
In particular, Aliakbarpour~\etal~\cite{AliakbarpourBGP18} presented an algorithm for estimating occurrences of any star $S_\ell$ for $\ell \geq 1$ using $\Os(\frac{m}{(\sS)^{1/\ell}})$ queries in our query model (including edge-sample queries) which is always at 
least as good as our bound in Theorem~\ref{thm:main-intro}, but potentially can be better. On the other hand, we show that our current algorithm, with almost no further modification, in fact achieves this stronger bound \emph{using a different analysis} (see 
Remark~\ref{rem:star-variance} for details). 

Additionally, as we pointed out before, our algorithm can solve
the more general database join size estimation for binary relations, or equivalently the subgraph counting problem with colors on edges. In Theorem~\ref{thm:lb-colorful}, we prove that for this more general problem, 
our algorithm in Theorem~\ref{thm:main-intro} indeed achieves optimal bounds for \emph{all choices} of the subgraph $H$.

\paragraph{Edge-Sample Queries.} The edge-sample query that we assume is not part of the standard access model for sublinear algorithms, namely the ``general graph'' query model (see, e.g.~\cite{KaufmanKR04}). Nonetheless, we find allowing for this query ``natural'' 
owing to the following factors: 
\vspace{-8pt}
\begin{enumerate}[leftmargin=10pt]
\item[] \emph{Theoretical implementation.} Edge sampling queries can be implemented with an $\Ot(n/\sqrt{m})$ multiplicative overhead in query and time using the recent result of~\cite{EdenR18}, or with an $O(n)$ additive 
preprocessing time (which is still sublinear in $m$) by querying degrees of all vertices. Hence, we can focus on designing algorithms by allowing these queries and later replacing them by either of the above implementations in a black-box way at a certain additional cost. 

\item[] \emph{Practical implementation.} Edge sampling is a common practice in analyzing social networks~\cite{LeskovecF06,LeeP15} or biological networks~\cite{AhmedNK13}. Another scenario when random edge sampling is possible
is when we can access a random location of the memory that is used to store the graph. To quote~\cite{AliakbarpourBGP18}: ``because edges normally take most of the space for storing graphs, an access to a random memory location where the 
adjacency list is stored, would readily give a random edge.'' Hence, assuming edge sampling queries can be considered valid in many scenarios. 

\item[] \emph{Understanding the power of random edge queries.} Edge sampling is a critical component of various sublinear time algorithms for graph estimation~\cite{EdenLRS15,EdenRS17,AliakbarpourBGP18,EdenRS18,EdenR18}. However,
except for~\cite{AliakbarpourBGP18} that also assumed edge-sample queries, all these other algorithms employ different workarounds to these queries. As we show in this paper, decoupling these workarounds from the rest of the algorithm by 
allowing edge-sample queries results in considerably simpler and more general algorithms for subgraph counting and is hence worth studying on its own. We also mention that studying the power of edge-sample queries has  been 
cast as an open question in~\cite{EdenR18b} as well.  
\end{enumerate}
\vspace{-8pt}

\paragraph{Implications to Streaming Algorithms.} Subgraph counting is also one of the most studied problems in the graph streaming model 
(see, e.g.~\cite{Bar-YossefKS02,JowhariG05,BuriolFLMS06,KaneMSS12,BravermanOV13,JhaSP13,SimpsonSM15,McGregorVV16,CormodeJ17,BeraC17} and references therein). In this model, the edges of the input graph are presented one by one in a stream; 
the algorithm makes a single or a small number of passes over the stream and outputs the answer after the last pass. The goal in this model is to minimize the memory used by the 
algorithm (similar-in-spirit to minimizing the query complexity in our query model). 

Our algorithm in Theorem~\ref{thm:main-intro} can be directly adapted to the streaming model, resulting in an algorithm for subgraph counting that makes $O(1)$ passes over the stream and uses
a memory of size $\Os\Paren{\min\set{m,\frac{m^{\rho(H)}}{\sH}}}$. For the case of counting triangles and cliques, the space complexity of our algorithm matches the best known
algorithms of McGregor~\etal~\cite{McGregorVV16} and Bera and Chakrabarti~\cite{BeraC17} which are known to be optimal~\cite{BeraC17}.  To the best of our knowledge, the only previous streaming algorithms for counting arbitrary subgraphs $H$ 
are those of Kane~\etal~\cite{KaneMSS12} and~Bera and Chakrabarti~\cite{BeraC17} that use, respectively, one pass and $\Os(\frac{m^{2\cdot\card{E(H)}}}{(\sH)^2})$ space, and two passes and $\Os(\frac{m^{\beta(H)}}{\sH})$ space, where $\beta(H)$ is the \emph{integral}
edge-cover number of $H$. As $\rho(H) \leq \beta(H) \leq \card{E(H)}$ by definition and $\sH \leq m^{\rho(H)}$ by AGM bound, the space complexity of our algorithm is always at least as good as the ones in~\cite{KaneMSS12,BeraC17} but potentially 
can be much smaller.  

\subsection{Main Ideas in Our Algorithm}\label{sec:techniques}

Our starting point is the AGM bound which implies that the number of ``potential copies'' of $H$ in $G$ is at most $m^{\rho(H)}$. Our goal of estimating $\sH$ then translates 
to counting how many of these potential copies  form an actual copy of $H$ in $G$. A standard approach at this point is the \emph{Monte Carlo method:}  sample a potential copy of $H$ in $G$ \emph{uniformly at random} and check
whether it forms an actual copy of $H$ or not; a simple exercise in concentration inequalities then implies that we only need $O(\frac{m^{\rho(H)}}{\sH})$ many \emph{independent} samples to get a good estimate of $\sH$. 

This approach however immediately runs into a technical difficulty. Given only a query access to $G$, it is not at all clear how to sample a potential copy of $H$ from the list of all potential copies. Our first task is then to design a procedure for sampling potential 
copies of $H$ from $G$. In order to do so, we again consider the AGM bound and the optimal fractional edge-cover that is used to derive this bound. We first prove a simple structural result that states that an optimal fractional edge-cover of $H$
can be supported only on edges that form a disjoint union of \emph{odd cycles} and \emph{stars} (in $H$). This allows us to decompose $H$ into a collection of odd cycles and stars and
treat any arbitrary subgraph $H$ as a collection of these simpler subgraphs that are suitably connected together.

The above decomposition reduces the task of sampling a potential copy of $H$ to sampling a collection of odd cycles and stars. Sampling an odd cycle $C_{2k+1}$ on $2k+1$ edges is as follows: 
sample $k$ edges $e_1,\ldots,e_k$ uniformly at random from $G$; pick one of the endpoints of $e_1$ and sample a vertex $v$ from the neighborhood of this endpoint uniformly at random. With some additional care, 
one can show that the tuple $(e_1,\ldots,e_k,v)$ sampled here is enough to identify an odd cycle of length $2k+1$ uniquely. To sample a star $C_{\ell}$ with $\ell$ petals, we sample a vertex $v$ from $G$ with probability proportional to its degree (by sampling a 
random edge and picking one of the two endpoints uniformly), and then sample $\ell$ vertices $w_1,\ldots,w_\ell$ from the neighborhood of $v$. Again, with some care, this allows us to sample 
a potential copy of a star $S_{\ell}$. We remark that these sampling procedures are related to sampling triangles in~\cite{EdenLRS15} and stars in~\cite{AliakbarpourBGP18}. Finally, to sample a potential copy of $H$, we simply sample all its odd cycles and stars 
in the decomposition using the method above. We should note right away that this however does \emph{not} result in a uniformly at random sample of potential copies of $H$ as various parameters of the graph $G$, in particular degrees of vertices, alter the 
probability of sampling each potential copy. 

The next and paramount step is then how to use the samples above to estimate the value of $\sH$. Obtaining an unbiased estimator of $\sH$ based on these samples is not hard as we can identify the probability each potential copy is sampled with in this process (which is a function of degrees of vertices of the potential copy in $G$) and reweigh each sample accordingly. Nevertheless, the variance of a vanilla variant of this sampling and reweighing approach is quite large for our purpose. To fix this, we use an idea 
similar to that of~\cite{EdenLRS15} for counting triangles: sample a ``partial'' potential copy of $H$ first and fix it; sample \emph{multiple} ``extensions'' of this partial potential copy to a complete potential copy and use the average of estimates
based on each extension to reduce the variance. More concretely, this translates to sampling multiple copies of the first cycle for the decomposition and for each sampled cycle, recursively sampling multiple copies of the remainder of $H$ as specified by the 
decomposition. A careful analysis of this recursive process---which is the main technical part of the paper---allows us to bound the variance of the estimator by $O(m^{\rho(H)}) \cdot (\sH)$. Repeating such an estimator 
$O(\frac{m^{\rho(H)}}{\sH})$ times independently and taking the average value then gives us a $(1 \pm \eps)$-approximation to $\sH$ by a simple application of Chebyshev's inequality.

\subsection{Further Related Work}\label{sec:related}

In addition to the previous work in~\cite{EdenLRS15,EdenRS18,AliakbarpourBGP18} that are already discussed above, sublinear-time algorithms for estimating subgraph counts and related parameters such as average degree and 
degree distribution moments have also been studied in~\cite{Feige04,GoldreichR08,GonenRS10,EdenRS17}. Similarly, sublinear-time algorithms are also studied for estimating other graph parameters such as weight of the minimum spanning
tree~\cite{CzumajS04,ChazelleRT05,CzumajEFMNRS05} or size of a maximum matching or a minimum vertex cover~\cite{ParnasR07,NguyenO08,YoshidaYI09,HassidimKNO09,OnakRRR12} (this is by no means a comprehensive summary of previous results). 

Subgraph counting has also been studied extensively in the graph streaming model 
(see, e.g.~\cite{Bar-YossefKS02,JowhariG05,BuriolFLMS06,KaneMSS12,BravermanOV13,JhaSP13,SimpsonSM15,McGregorVV16,CormodeJ17,BeraC17,KP17,KKP18} and references therein). In this model, the edges of the input graph are presented one by one in a stream; 
the algorithm makes a single or a small number of passes over the stream and outputs the answer after the last pass. The goal in this model is to minimize the memory used by the 
algorithm similar-in-spirit to minimizing the query complexity in our query model. However, the streaming algorithms typically require reading the entire graph in the stream which is different from our goal in sublinear-time algorithms.

\section{Preliminaries}\label{sec:prelim}

\paragraph{Notation.} For any integer $t \geq 1$, we let $[t] := \set{1,\ldots,t}$. For any event $\event$, $\II(\event) \in \set{0,1}$ is an indicator denoting whether $\event$ happened or not. 
For a graph $G(V,E)$, $V(G) := V$ denotes the vertices and $E(G):= E$ denotes the edges. For a vertex $v \in V$, $N(v)$ denotes the neighbors of $v$, and $d_v := \card{N(v)}$ denotes the degree of $v$. 

To any edge $e = \set{u,v}$ in $G$, we assign two {directed} edges $\ve_1 = (u,v)$ and $\ve_2 = (v,u)$ called the directed copies of $e$ and let $\vE$ be the set of all these directed edges. 
We also fix a total ordering $\prec$ on vertices whereby for any two vertices $u,v \in V$, $u \prec v$ iff $d_u < d_v$, or $d_u = d_v$ and $u$ appears before $v$ in the lexicographic order.  
To avoid confusion, we use letters $a,b$ and $c$ 
to denote the vertices in the subgraph $H$, and letters $u,v$ and $w$ to denote the vertices of $G$.

We use the following standard variant of Chebyshev's inequality. 

\begin{proposition}\label{prop:cheb}
	For any random variable $X$ and integer $t \geq 1$, 
	$
		\Pr\paren{\card{X - \Ex\bracket{X}} \geq t} \leq \frac{\var{X}}{t^2}.
	$
\end{proposition}
\noindent
We also recall the law of total variance that states the for two random variables $X$ and $Y$, 
\begin{align}
	\var{Y} = \Ex_x\paren{\var{Y \mid X=x}} + \Var_x\bracket{\expect{Y \mid X=x}}. \label{eq:total-variance}
\end{align}

\noindent
We use the following standard graph theory fact in our proofs (see Appendix~\ref{app:min-degree} for a proof). 

\begin{proposition}[cf.~\cite{ChibaN85}]\label{prop:min-degree}
	For any graph $G(V,E)$, $\sum_{(u,v) \in E} \min(d_u,d_v) \leq 5m\sqrt{m}$. 
\end{proposition}

\paragraph{Assumption on Size of Subgraph $H$.} Throughout the paper, we assume that the size of the subgraph $H$ is a fixed constant independent of the size of the graph $G$ and hence we suppress the dependency on size of $H$
in various bounds in our analysis using $O$-notation.

\newcommand{\tsH}{\ensuremath{\widetilde{\sH}}}

\section{A Graph Decomposition Using Fractional Edge-Covers}\label{sec:edge-cover}

In this section, we give a simple decomposition of the subgraph $H$ using fractional edge-covers. We start by defining fractional edge-covers formally (see also Figure~\ref{fig:decomposition}). 
\begin{definition}[Fractional Edge-Cover Number] A {fractional edge-cover} of $H(V_H,E_H)$ is a mapping $\psi : E_H \rightarrow [0,1]$ such that for each vertex $a \in V_H$, $\sum_{e \in E_H, a \in e} \psi(e) \geq 1$. 
The fractional edge-cover number $\rho(H)$ of $H$ is the minimum value of $\sum_{e \in E_H} \psi(e)$ among all fractional edge-covers $\psi$. 
\end{definition}

\noindent
The fractional edge-cover number of a graph can be computed by the following LP:
\begin{align}
\rho(H)\quad = \quad &\textnormal{minimize} && \hspace{-2cm} \textnormal{$\sum_{e \in E(H)}x_e$} \notag \\
&\textnormal{subject to} && \hspace{-2.05cm} \textnormal{$\sum_{e \in E_H: a \in e} x_e \geq 1$ for all vertices $a \in V(H)$.} \label{lp:ec}
\end{align}

The following lemma is the key to our decomposition. The proof is based on standard ideas in linear programming and is provided in Appendix~\ref{app:subgraph-decomposition} for completeness. 

\begin{lemma}\label{lem:subgraph-decomposition}
	Any subgraph $H$ admits an optimal fractional edge-cover $x^*$ such that the support of $x^*$, denoted by $\supp{x^*}$, is a collection of vertex-disjoint
	odd cycles and star graphs, and, 
	\begin{enumerate}
		\item for every odd cycle $C \in \supp{x^*}$, $x^*_e = 1/2$ for all $e \in C$;
		\item for every edge $e \in \supp{x^*}$ that does \emph{not} belong to any odd cycle,  $x_e = 1$. 
	\end{enumerate}
\end{lemma}

\subsection{The Decomposition}

We now present the decomposition of $H$ using Lemma~\ref{lem:subgraph-decomposition}. Let $x^*$ be an optimal fractional edge-cover in Lemma~\ref{lem:subgraph-decomposition} and let $\CC_1,\ldots,\CC_o$ be the odd-cycles
in the support of $x^*$ and $\SS_1,\ldots,\SS_s$ be the stars. We define $\DD(H) := \set{\CC_1,\ldots,\CC_o,\SS_1,\ldots,\SS_s}$ as the decomposition of $H$ (see Figure~\ref{fig:decomposition} below for an illustration). 

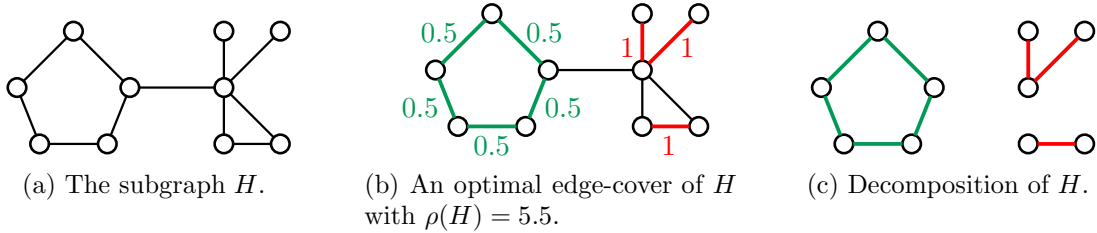
\begin{figure}[h!]
    \centering
    \subcaptionbox{The subgraph $H$.}[0.3\textwidth]{

    \begin{tikzpicture}[ auto ,node distance =1cm and 2cm , on grid , semithick , state/.style ={ circle ,top color =white , bottom color = white , draw, black , text=black}, every node/.style={inner sep=0,outer sep=0}]

\node[state, circle, black, line width=0.35mm, minimum height=7pt, minimum width=7pt] (a1){};
\node[state, circle, black, line width=0.35mm, minimum height=7pt, minimum width=7pt] (a2) [below left=0.75cm and 0.75cm of a1]{};
\node[state, circle, black, line width=0.35mm, minimum height=7pt, minimum width=7pt] (a3) [below right=0.75cm and 0.75cm of a1]{};
\node[state, circle, black, line width=0.35mm, minimum height=7pt, minimum width=7pt] (a4) [below right=0.75cm and 0.3cm of a2]{};
\node[state, circle, black, line width=0.35mm, minimum height=7pt, minimum width=7pt] (a5) [below left=0.75cm and 0.3cm of a3]{};

\node[state, circle, black, line width=0.35mm, minimum height=7pt, minimum width=7pt](b1) [below right=0.75cm and 2cm of a1]{};
\node[state, circle, black, line width=0.35mm, minimum height=7pt, minimum width=7pt] (b2) [below=0.75cm of b1]{};
\node[state, circle, black, line width=0.35mm, minimum height=7pt, minimum width=7pt](b3) [right=0.75cm of b2]{};
\node[state, circle, black, line width=0.35mm, minimum height=7pt, minimum width=7pt] (b4) [above=0.75cm of b1]{};
\node[state, circle, black, line width=0.35mm, minimum height=7pt, minimum width=7pt] (b5) [right=0.75cm of b4]{};

\draw[line width=0.3mm](a1) to (a2);
\draw[line width=0.3mm] (a1) to (a3);
\draw[line width=0.3mm] (a2) to (a4);
\draw[line width=0.3mm] (a4) to (a5);
\draw[line width=0.3mm] (a5) to (a3);

\draw[line width=0.3mm] (b1) to (b2);
\draw[line width=0.3mm] (b1) to (b3);
\draw[line width=0.3mm] (b2) to (b3);
\draw[line width=0.3mm] (b1) to (b4);
\draw[line width=0.3mm] (b1) to (b5);

\draw[line width=0.3mm] (a3) to (b1);

\end{tikzpicture}

}
\hspace{0.2cm}   
    \subcaptionbox{An optimal edge-cover of $H$ with $\rho(H) = 5.5$.}[0.3\textwidth]{

    \begin{tikzpicture}[ auto ,node distance =1cm and 2cm , on grid , semithick , state/.style ={ circle ,top color =white , bottom color = white , draw, black , text=black}, every node/.style={inner sep=0,outer sep=0}]

\node[state, circle, black, line width=0.35mm, minimum height=7pt, minimum width=7pt] (a1){};
\node[state, circle, black, line width=0.35mm, minimum height=7pt, minimum width=7pt] (a2) [below left=0.75cm and 0.75cm of a1]{};
\node[state, circle, black, line width=0.35mm, minimum height=7pt, minimum width=7pt] (a3) [below right=0.75cm and 0.75cm of a1]{};
\node[state, circle, black, line width=0.35mm, minimum height=7pt, minimum width=7pt] (a4) [below right=0.75cm and 0.3cm of a2]{};
\node[state, circle, black, line width=0.35mm, minimum height=7pt, minimum width=7pt] (a5) [below left=0.75cm and 0.3cm of a3]{};

\node[state, circle, black, line width=0.35mm, minimum height=7pt, minimum width=7pt](b1) [below right=0.75cm and 2cm of a1]{};
\node[state, circle, black, line width=0.35mm, minimum height=7pt, minimum width=7pt] (b2) [below=0.75cm of b1]{};
\node[state, circle, black, line width=0.35mm, minimum height=7pt, minimum width=7pt](b3) [right=0.75cm of b2]{};
\node[state, circle, black, line width=0.35mm, minimum height=7pt, minimum width=7pt] (b4) [above=0.75cm of b1]{};
\node[state, circle, black, line width=0.35mm, minimum height=7pt, minimum width=7pt] (b5) [right=0.75cm of b4]{};

\draw[line width=0.3mm](a1) to (a2);
\draw[line width=0.3mm] (a1) to (a3);
\draw[line width=0.3mm] (a2) to (a4);
\draw[line width=0.3mm] (a4) to (a5);
\draw[line width=0.3mm] (a5) to (a3);

\draw[line width=0.3mm] (b1) to (b2);
\draw[line width=0.3mm] (b1) to (b3);
\draw[line width=0.3mm] (b2) to (b3);
\draw[line width=0.3mm] (b1) to (b4);
\draw[line width=0.3mm] (b1) to (b5);

\draw[line width=0.3mm] (a3) to (b1);

\draw[line width=0.5mm, ForestGreen] (a1) to (a2) node[above right=0.5cm and 0.05cm of a2] {$0.5$};
\draw[line width=0.5mm, ForestGreen] (a1) to (a3) node[above left=0.5cm and 0.05cm of a3] {$0.5$};
\draw[line width=0.5mm, ForestGreen] (a2) to (a4) node[above left=0.25cm and 0.5cm of a4] {$0.5$};
\draw[line width=0.5mm, ForestGreen] (a4) to (a5) node[below left=0.25cm and 0.45cm of a5] {$0.5$};
\draw[line width=0.5mm, ForestGreen] (a5) to (a3) node[above right=0.25cm and 0.5cm of a5] {$0.5$};

\draw[line width=0.5mm, red] (b1) -- (b4) node [above left=0.3cm and 0.2cm of b1] {$1$};
\draw[line width=0.5mm, red] (b1) -- (b5) node [above right=0.3cm and 0.6cm of b1] {$1$};
\draw[line width=0.5mm, red] (b3) -- (b2) node[below right=0.25cm and 0.35cm of b2] {$1$};

\end{tikzpicture}

}\hspace{0.2cm}
\subcaptionbox{Decomposition of $H$.}[0.3\textwidth]{

    \begin{tikzpicture}[ auto ,node distance =1cm and 2cm , on grid , semithick , state/.style ={ circle ,top color =white , bottom color = white , draw, black , text=black}, every node/.style={inner sep=0,outer sep=0}]

\node[state, circle, black, line width=0.35mm, minimum height=7pt, minimum width=7pt] (a1){};
\node[state, circle, black, line width=0.35mm, minimum height=7pt, minimum width=7pt] (a2) [below left=0.75cm and 0.75cm of a1]{};
\node[state, circle, black, line width=0.35mm, minimum height=7pt, minimum width=7pt] (a3) [below right=0.75cm and 0.75cm of a1]{};
\node[state, circle, black, line width=0.35mm, minimum height=7pt, minimum width=7pt] (a4) [below right=0.75cm and 0.3cm of a2]{};
\node[state, circle, black, line width=0.35mm, minimum height=7pt, minimum width=7pt] (a5) [below left=0.75cm and 0.3cm of a3]{};

\node[state, circle, black, line width=0.35mm, minimum height=7pt, minimum width=7pt](b1) [below right=0.75cm and 2cm of a1]{};
\node[state, circle, black, line width=0.35mm, minimum height=7pt, minimum width=7pt] (b2) [below=0.75cm of b1]{};
\node[state, circle, black, line width=0.35mm, minimum height=7pt, minimum width=7pt](b3) [right=0.75cm of b2]{};
\node[state, circle, black, line width=0.35mm, minimum height=7pt, minimum width=7pt] (b4) [above=0.75cm of b1]{};
\node[state, circle, black, line width=0.35mm, minimum height=7pt, minimum width=7pt] (b5) [right=0.75cm of b4]{};

\draw[line width=0.5mm, ForestGreen] (a1) to (a2);
\draw[line width=0.5mm, ForestGreen] (a1) to (a3);
\draw[line width=0.5mm, ForestGreen] (a2) to (a4);
\draw[line width=0.5mm, ForestGreen] (a4) to (a5);
\draw[line width=0.5mm, ForestGreen] (a5) to (a3);

\draw[line width=0.5mm, red] (b1) -- (b4);
\draw[line width=0.5mm, red] (b1) -- (b5);
\draw[line width=0.5mm, red] (b3) -- (b2);

\end{tikzpicture}

}
    
    \caption{Illustration of the our decomposition for $H$ based on fractional edge-covers.}
    \label{fig:decomposition}
\end{figure}

For every $i \in [o]$, let the length of the odd cycle $\CC_i$ be $2k_i+1$ (i.e., $\CC_i = C_{2k_i+1}$); we define $\rhoC_{i} := k_i+1/2$. Similarly, for every $j \in [s]$, let the number of petals in $\SS_j$ be $\ell_j$ (i.e., $\SS_j = S_{\ell_j}$); we define
$\rhoS_j := \ell_j$. By Lemma~\ref{lem:subgraph-decomposition}, 
\begin{align}
	\rho(H) = \sum_{i=1}^{o} \rhoC_i + \sum_{j=1}^{s} \rhoS_j \label{eq:agm}.
\end{align}
\noindent
Recall that by AGM bound, the total number of copies of $H$ possible in $G$ is $m^{\rho(H)}$. 
We also use the following simple lemma which is a direct corollary of the AGM bound.

\begin{lemma}\label{lem:subgraph-AGM}
	Let $I := \set{i_1,\ldots,i_o}$ and $J:=\set{j_1,\ldots,j_s}$ be subsets of $[o]$ and $[s]$, respectively. Suppose 
	$\tH$ is the subgraph of $H$ on vertices of the odd cycles $\CC_{i_1},\ldots,\CC_{i_o}$ and stars $\SS_{j_1},\ldots,\SS_{j_s}$. 
	Then the total number of copies of $\tH$ in $G$ is at most $m^{\rho(\tH)}$ for $\rho(\tH) \leq \sum_{i \in I} \rhoC_i + \sum_{j \in J} \rhoS_j$.
\end{lemma}
\begin{proof}
	Let $x^*$ denote the optimal value of LP~(\ref{lp:ec}) in the decomposition $\DD(H)$. Define $y^*$ as the projection of $x^*$ to edges present in $\tH$. It is easy to see that $y^*$ is a feasible solution for
	LP~(\ref{lp:ec}) of $\tH$ with value $\sum_{i \in I} \rhoC_i + \sum_{j \in J} \rhoS_j$. The lemma now follows from  the AGM bound for $\tH$. 
\end{proof}

\subsection{Profiles of Cycles, Stars, and Subgraphs} \label{sec:profiles}

We conclude this section by specifying the representation of the potential occurrences of the subgraph $H$ in $G$ based on the decomposition $\DD(H)$. 

\noindent
\emph{Odd cycles:} We represent a potential occurrence of an odd cycle $C_{2k+1}$ in $G$ as follows. 
Let $\be= (\ve_1,\ldots,\ve_k) \in \vE^k$ be an ordered tuple of $k$ directed copies of edges in $G$ and suppose $\ve_i := (u_i,v_i)$ for all $i \in [k]$. Define $\ustar_{\be} 
= u_1$ and let $w$ be any vertex in $N(\ustar_{\be})$. 
We refer to any such collection $(\be,w)$ as a \emph{profile} of $C_{2k+1}$ in $G$. We say that ``the profile $(\be,w)$ {forms} a cycle $C_{2k+1}$ in $G$'' 
iff $(i)$  $u_1$ is the smallest vertex on the cycle according to $\prec$, $(ii)$ $v_1 \prec w$, \emph{and} $(iii)$ 
the edges $(u_1,v_1),(v_1,u_2),\ldots,(u_k,v_k),(v_k,w),(w,u_1)$ all exist in $G$ and hence there is a copy of $C_{2k+1}$ on vertices $\set{u_1,v_1,u_2,v_2,\ldots,u_k,v_k,w}$ in $G$. 
Note that under this definition and our definition of $\sC$, each copy of $C_{2k+1}$ correspond to exactly one profile $(\be,w)$ and vice versa. As such,
\begin{align}
	\sC = \sum_{\be \in \vE^k} \sum_{w \in N(\ustar_{\be})} \II\Paren{\textnormal{$(\be,w)$ forms a cycle $C_{2k+1}$ in $G$}}. \label{eq:cycle-profile}
\end{align}
\noindent	
\emph{Stars:} We represent a potential occurrence of a star $S_{\ell}$ in $G$ by $(v,\bw)$ where $v$ is the center of the star and $\bw =(w_1,\ldots,w_\ell)$ are the $\ell$ petals. We refer to $(v,\bw)$ as a \emph{profile} of $S_\ell$ in $G$. 
We say that ``the profile $(v,\bw)$ forms a star $S_{\ell}$ in $G$'' iff $(i)$ $\card{\bw} > 1$, \emph{or} $(ii)$ $(\ell=)\card{\bw} = 1$ and $v \prec w_1$; in both cases there is a copy of $S_\ell$ on  vertices $v,w_1,\ldots,w_{\ell}$.
Under this definition, each copy of $S_\ell$ corresponds to exactly one profile $(v,\bw)$. As such, 
\begin{align}
	\sS = \sum_{v \in V} \sum_{\bw \in N(v)^{\ell}} \II\Paren{\textnormal{$(v,\bw)$ forms a star $S_\ell$ in $G$}}. \label{eq:star-profile}
\end{align}
\noindent
\emph{Arbitrary subgraphs:} We represent a potential occurrence of $H$ in $G$ by an $(o+s)$-tuple 
$\bR:= ((\be_1,w_1),\ldots,(\be_o,w_o),(v_1,\bw_1),\ldots,(v_s,\bw_s))$ where $(\be_i,w_i)$ is a profile of the cycle $\CC_i$ in $\DD(H)$ and $(v_j,\bw_j)$ is a profile
of the star $\SS_j$. We refer to $\bR$ as a \emph{profile} of $H$ and say that ``the profile $\bR$ forms a copy of $H$ in $G$'' iff $(i)$ each profile forms a corresponding copy of $\CC_i$ or $\SS_j$ in $\DD(H)$, and $(ii)$ the remaining 
edges of $H$ between vertices specified by $\bR$ all are present in $G$ (note that by definition of the decomposition $\DD(H)$, all vertices of $H$ are specified by $\bR$). As such,
\begin{align}
	\sH = \sum_{\bR} \II\Paren{\textnormal{$\bR$ forms a copy of $H$ in $G$}} \cdot f(H), \label{eq:subgraph-profile}
\end{align}
for a fixed constant $f(H)$ depending only on $H$ as defined below. 
Let $\pi: V_H\to V_H$ be an automorphism of $H$. Let $C_1,\ldots, C_o$, $S_1,\ldots, S_s$ denote the cycles and stars in the decomposition of $H$. We say that $\pi$ is decomposition preserving if for every $i=1,\ldots, o$ cycle $C_i$ is mapped to a cycle of the 
same length and for every $i=1,\ldots, s$ $S_i$ is mapped to a star with the same number of petals. Let the number of decomposition preserving automorphisms of $H$ be denoted by $Z$, and define $f(H)=1/Z$. Define the quantity 
$\tsH := \sum_{\bR} \II\Paren{\textnormal{$\bR$ forms a copy of $H$ in $G$}}$ which is equal to $\sH$ modulo the scaling factor of $f(H)$. It is immediate that estimating $\sH$ and $\tsH$ are equivalent to each other and hence 
in the rest of the paper, with a slight abuse of notation, we use $\sH$ and $\tsH$ interchangeably.



\section{A Sublinear-Time Algorithm for Subgraph Counting}\label{sec:alg}

We now present our sublinear time algorithm for approximately counting number of any given arbitrary subgraph $H$ in an underlying graph $G$ and prove Theorem~\ref{thm:main-intro}. 
The main component of our algorithm is an unbiased estimator random variable for $\sH$ with low variance. The algorithm in Theorem~\ref{thm:main-intro} is then obtained by 
simply repeating this unbiased estimator in parallel enough number of times
(based on the variance) and outputting the average value of these estimators. 

\subsection{A Low-variance Unbiased Estimator for $\sH$}\label{sec:estimator}

We present a low-variance unbiased estimator for $\sH$ in this section. Our algorithm is a sampling based algorithm. In the following, we first introduce
two separate subroutines for sampling odd cycles ($\scestimator$) and stars ($\ssestimator$), and then use these components in conjunction with the decomposition we introduced in Section~\ref{sec:edge-cover}, 
to present our full algorithm. We should  right away clarify that $\scestimator$ and $\ssestimator$ are not exactly sampling a cycle or a star, but rather sampling a set of vertices and edges (in a non-uniform way) that can potentially
form a cycle or star in $G$, i.e., they sample a profile of these subgraphs defined in Section~\ref{sec:profiles}. 

\subsubsection*{The $\scestimator$ Algorithm}

We start with the following algorithm for sampling an odd cycle $C_{2k+1}$ for some $k \geq 1$. 
{This algorithm outputs a simple data structure, named the \emph{cycle-sampler tree}, that provides a convenient representation of the samples taken by our algorithm (see Definition~\ref{def:cst} immediately after the description of the algorithm). This data structure can be easily avoided when designing a cycle counting algorithm, but will be quite useful for reasoning about the recursive structure of our sampling algorithm for general graphs $H$.}

\begin{tbox}
	$\scestimator(G,C_{2k+1})$. 
	
	\begin{enumerate}
		\item\label{line1:sample-query} Sample $k$ directed edges $\be := \paren{\ve_1,\ldots,\ve_{k}}$ uniformly at random (with replacement) from $G$
		with the constraint that for $\ve_1 = (u_1,v_1)$, $u_1 \prec v_1$. 
		\item\label{line1:degree-query} Let $\ustar_{\be} := u_1$  and let $\dstar_{\be} := d_{\ustar_{\be}}$. 
		\item\label{line1:for-loop} For $i=1$ to $t_{\be}:= \ceil{{\dstar_{\be}}/{\sqrt{m}}}$:  Sample a vertex $w_i$ uniformly at random from $N(\ustar_{\be})$. 
	\item Let $\bw := (w_1,\ldots,w_{t_{\be}})$. Return the cycle-sampler tree $\TT(\be,\bw)$ (see Definition~\ref{def:cst}).
	\end{enumerate}
\end{tbox}

 \begin{definition}[Cycle-Sampler Tree]\label{def:cst}
 	The cycle-sampler tree $\TT(\be,\bw)$ for the tuple $(\be,\bw)$ sampled by $\scestimator(G,C_{2k+1})$ is the following 2-level tree $\TT$: 
	\begin{itemize}
		\item Each node $\alpha$ of the tree contains two attributes: $\Label{\alpha}$ which consists of some of the edges and vertices in $(\be,\bw)$, and an integer $\Value{\alpha}$.
		\item For the root $\alpha_r$ of $\TT$, $\Label{\alpha_r} := \be$ and $\Value{\alpha_r} := (2m)^{k}/2$. 
		
		\emph{(}$\Value{\alpha_r}$ is equal to the inverse
		of the probability that $\be$ is sampled by $\scestimator$\emph{)}. 
		\item The root $\alpha_r$ has $t_{\be}$ child-nodes in $\TT$ for a parameter $t_{\be}= \ceil{{\dstar_{\be}}/{\sqrt{m}}}$ (consistent with line~3 of $\scestimator(G,C_{2k+1})$ above). 
		\item For the $i$-th child-node $\alpha_i$ of root, $i\in [t_{\be}]$, $\Label{\alpha_i} := w_i$ and $\Value{\alpha_i} := \dstar_{\be}$ 
		
		\emph{(}$\Value{\alpha_i}$ is equal to the inverse
		of the probability that $w_i$ is sampled by $\scestimator$, conditioned on $\be$ being sampled\emph{)}.
	\end{itemize}
	Moreover, for each root-to-leaf path $\PP_i := (\alpha_r,\alpha_i)$ (for $i \in [t_{\be}]$), define $\Label{\PP_i} := \Label{\alpha_r} \cup \Label{\alpha_i}$ and $\Value{\PP_i} := \Value{\alpha_r} \cdot \Value{\alpha_i}$ 
	\emph{(}$\Label{\PP_i}$ is a profile of the cycle $C_{2k+1}$ as defined in Section~\ref{sec:profiles}\emph{)}.
 \end{definition}
 \noindent
 See Figure~\ref{fig:A} for an illustration of a cycle-sampler tree.  
 
$\scestimator$ can be implemented in our query model by using $k$ edge-sample queries (and picking the correct direction for $e_1$ based on $\prec$ and one of the two directions uniformly at random for the other edges) in Line~(\ref{line1:sample-query}), two degree queries in Line~(\ref{line1:degree-query}), and one neighbor query in Line~(\ref{line1:for-loop}). This
results in $O(k)$ queries in total for one iteration of the for-loop in Line~(\ref{line1:for-loop}). As such, the total query complexity of $\scestimator$ is $O(t_{\be})$ (recall that $k$ is a constant). It is also 
straightforward to verify that we can compute the cycle-sampler tree $\TT$ of an execution of $\scestimator$ with no further queries and in $O(t_{\be})$ time. 
We bound the query complexity of this algorithm by bounding the expected number of iterations in the for-loop. 

 \begin{lemma}\label{lem:scestimator-time}
 	For the parameter $t_{\be}$ in Line~(\ref{line1:for-loop}) of $\scestimator$, $\Ex\bracket{t_{\be}} = O(1)$. 
 \end{lemma}
 \begin{proof}
 	By definition, $t_{\be} := \ceil{{\dstar_{\be}}/{\sqrt{m}}}$ for $\dstar_{\be} = \min(d_u,d_v)$ for an edge $e_1 = (u,v)$ chosen uniformly at random from $G$. As such, by Proposition~\ref{prop:min-degree}, 
	\begin{align*}
		\Ex\bracket{t_{\be}} = 1+O(m^{-1/2}) \cdot \frac{1}{m} \cdot \sum_{(u,v) \in E} \min(d_u,d_v) \Eq{Proposition~\ref{prop:min-degree}} O(m^{-3/2}) \cdot 5m\sqrt{m}= O(1). \qed
	\end{align*}
	
 \end{proof}

We now define a process for estimating the number of odd cycles in a graph using the information stored in the cycle-sampler tree and the $\scestimator$ algorithm. 
While we do not use this process in a black-box way in our main algorithm, abstracting it out makes the analysis of our main algorithm simpler to follow and more transparent, and serves as a warm-up for our main algorithm.

\paragraph{Warm-up: An Estimator for Odd Cycles.} Let $\TT:=\scestimator(G, C_{2k+1})$ be the output of an invocation of $\scestimator$. Note that the cycle-sampler tree $\TT$ is a random variable depending on the randomness of $\scestimator$. We define the 
random variable $X_i$ such that $X_i := \Label{\PP_i}$ for the $i$-th root-to-leaf path iff $\Label{\PP_i}$ forms a copy of $C_{2k+1}$ in $G$ and otherwise $X_i := 0$ (according to the definition of Section~\ref{sec:edge-cover}). 
We further define $Y := \frac{1}{t_{\be}} \cdot \sum_{i=1}^{t_{\be}} X_i$ (note that $t_{\be}$ is also a random variable). Our estimator algorithm can compute the value of these random variables
using the information stored in the tree $\TT$ plus additional $O(k) = O(1)$ queries for each of the $t_{\be}$ root-to-leaf path $\PP_i$ to detect whether $(\be,w_i)$ forms a copy of $H$ or not. Thus, the query
complexity and runtime of the estimator is still $O(t_{\be})$  (which in expectation is $O(1)$ by Lemma~\ref{lem:scestimator-time}). 
We now analyze its expectation and variance. 

 \begin{lemma}\label{lem:scestimator-output}
 	For the random variable $Y$ associated with $\scestimator(G,C_{2k+1})$, 
	\begin{align*}
		\expect{Y} &= (\sC), \qquad \var{Y} \leq (2m)^{k}\sqrt{m} \cdot \expect{Y}.
	\end{align*}
 \end{lemma}
 \begin{proof}
 	We first analyze the expected value of $X_i$'s and then use this to bound $\expect{Y}$. For any $i \in [t_{\be}]$, we have, $
	X_i = \Value{\alpha_r} \cdot \Value{\alpha_i} \cdot \II\paren{\text{$\Label{\alpha_r}\cup \Label{\alpha_i}$ forms a copy of $C_{2k+1}$}}$, where, as per Definition~\ref{def:cst}, $\alpha_r$ is the root of $\TT$.
	As such, 
	\begin{align*}
		\Ex\bracket{X_i} &= \sum_{\be \in \vE^{k}} \sum_{w \in N(\ustar_{\be})} \Pr\paren{\Label{\alpha_r} = \be} \cdot \Pr\paren{\Label{\alpha_i} = w} \\
		&\hspace{4cm} \cdot \II\paren{(\be,w) \textnormal{ forms a copy of $C_{2k+1}$}} \cdot \Value{\alpha_r} \cdot \Value{\alpha_i} \\ 
		&= \frac{2}{(2m)^{k}} \cdot \sum_{\be} \paren{\frac{1}{\dstar_{\be}} \cdot \sum_{w} \II\paren{(\be,w) \textnormal{ forms a copy of $C_{2k+1}$}} \cdot \frac{(2m)^{k}}{2} \cdot \dstar_{\be}} \\
		&= \sum_{\be} \sum_{w} \II\paren{(\be,w) \textnormal{ forms a copy of $C_{2k+1}$}} \Eq{Eq~(\ref{eq:cycle-profile})} (\sC).
	\end{align*}
	As $\expect{Y} = \expect{X_i}$ for any $i \in [t]$ by linearity of expectation, we obtain the desired bound on $\expect{Y}$. 
	
	We now bound $\var{Y}$ using the fact that it is obtained by taking 
	average of $t_{\be}$ random variables that are independent \emph{after} we condition on the choice of $\be$. We formalize this as follows. 
	Note that for any two $i \neq j$, the random variables $X_i \mid \be$ and $X_j \mid \be$
	are independent of each other (even though $X_i$ and $X_j$ in general are correlated). By the law of total variance in Eq~(\ref{eq:total-variance}),
	\begin{align}
		\var{Y} &= \expect{\var{Y \mid \be}} + \var{\expect{Y \mid \be}}. \label{eq:var1}
	\end{align}
	We bound each term separately now. Recall that $Y := \frac{1}{t_{\be}} \sum_{i=1}^{t_{\be}} X_i$ and $X_i$'s are independent conditioned on $\be$. As such, 
	\begin{align*}
		\expect{\var{Y \mid \be}} &= \frac{2}{(2m)^{k}} \sum_{\be \in \vE^k} \var{Y \mid \be} = \frac{2}{(2m)^{k}} \sum_{\be} \frac{1}{t^2_{\be}} \cdot \sum_{i=1}^{t_{\be}} \var{X_i \mid \be} \tag{by conditional independence of $X_i$'s} \\
		&\leq\frac{2}{(2m)^{k}} \sum_{\be} \frac{1}{t_{\be}} \expect{X_1^2 \mid \be} \tag{as distribution of all $X_i$'s are the same}.
	\end{align*}
	Hence, it suffices to calculate $\expect{X_1^2 \mid \be}$. We have, 
	\begin{align*}
		\expect{X^2_1 \mid \be} &= \sum_{w \in N(\ustar)} \Pr\paren{\Label{\alpha_1} = w} \cdot \II\paren{(\be,w) \textnormal{ forms a copy of $C_{2k+1}$}} \cdot \paren{\Value{\alpha_r} \cdot \Value{\alpha_1}}^2 \\
		&= \frac{1}{\dstar_{\be}} \cdot \sum_{w} \II\paren{(\be,w) \textnormal{ forms a copy of $C_{2k+1}$}} \cdot \paren{\frac{(2m)^{k}}{2} \cdot \dstar_{\be}}^2 \\
		&\leq \frac{(2m)^{2k}}{4} \cdot \dstar_{\be} \cdot \sum_{w} \II\paren{(\be,w) \textnormal{ forms a copy of $C_{2k+1}$}} \\
		&\hspace{-8pt}\Eq{Eq~(\ref{eq:cycle-profile})} \frac{(2m)^{2k}}{4} \cdot \dstar_{\be}  \cdot \ssC{\be},
	\end{align*}
	where $\ssC{\be}$ denotes the number of copies of $C_{2k+1}$ containing the sub-profile $\be = (\ve_1,\ldots,\ve_k)$. By plugging in this bound in the above equation, we have,
	\begin{align}
		\expect{\var{Y \mid \be}}  &\leq \frac{2}{(2m)^k} \sum_{\be \in \vE^k} \frac{1}{t_{\be}} \cdot \frac{(2m)^{2k}}{4} \cdot \dstar_{\be}  \cdot \ssC{\be} \notag \\
		&\leq \frac{(2m)^{k}\sqrt{m}}{2}\sum_{\be} \ssC{\be} = \frac{(2m)^{k}\sqrt{m}}{2} \cdot (\sC), \label{eq:var2}
	\end{align}
	by the choice of $t_{\be} \geq \dstar_{\be}/\sqrt{m}$. We now bound the second term in Eq~(\ref{eq:var1}). Note that by 
	the by proof of $\expect{Y} = (\sC)$ above, we also have $\expect{Y \mid \be} = \ssC{\be}$. As such, 
	\begin{align*}
		\var{\expect{Y \mid \be}}  &= \var{\ssC{\be}} \leq \expect{\ssC{\be}^2} = \frac{2}{(2m)^{k}} \sum_{\be} \ssC{\be}^2 \\
		&\leq \frac{2}{(2m)^{k}} \cdot \paren{\sum_{\be}\ssC{\be}}^2 = \frac{2}{(2m)^{k}} \paren{\sC}^2 
		\leq \sqrt{m} \cdot \paren{\sC},
	\end{align*}
	where the last inequality is by AGM bound in Lemma~\ref{lem:subgraph-AGM} which states that $\paren{\sC} \leq m^{k}\sqrt{m}$ (as $\rho(C_{2k+1}) = k+1/2$). Plugging in this bound in the second term of Eq~(\ref{eq:var1}) and using Eq~(\ref{eq:var2}) for the first term yields:
	\begin{align*}
		\var{Y} &= \expect{\var{Y \mid \be}} + \var{\expect{Y \mid \be}} \\
		&\leq \frac{(2m)^{k}\sqrt{m}}{2} \cdot (\sC) + \sqrt{m} \cdot \paren{\sC} \leq (2m)^{k}\sqrt{m} \cdot (\sC) = (2m)^{k} \sqrt{m} \cdot \expect{Y},
	\end{align*}
	where the equality is by the bound on $\expect{Y}$ proven in the first part. 
 \end{proof}

 \subsubsection*{The $\ssestimator$ Algorithm}
 
 We now give an algorithm for sampling a star $S_{\ell}$ with $\ell$ petals. 
 {Similar to $\scestimator$, this algorithm also outputs a simple data structure, named the \emph{star-sampler tree}, that provides a convenient representation of the samples taken by our algorithm (see Definition~\ref{def:sst}, immediately after the description of the algorithm). This data structure can be easily avoided when designing a star counting algorithm, but will be quite useful for reasoning about the recursive structure of our sampling algorithm for general graphs $H$.}

 \begin{tbox}
	$\ssestimator(G,S_\ell)$. 
	
	\begin{enumerate}
		\item\label{line2:sample-query} Sample a vertex $v \in V$ chosen with probability proportional to its degree in $G$ (i.e., for any vertex $u \in V$, $\Pr\paren{\text{$u$ is chosen as the vertex $v$}} = d_u/2m$).  
		\item\label{line2:neighbor-query} Sample $\ell$ vertices $\bw:=(w_1,\ldots,w_{\ell})$ from $N(v)$ uniformly at random (without replacement). 
		\item Return the star-sampler tree $\TT(v,\bw)$ (see Definition~\ref{def:sst}). 
	\end{enumerate}
\end{tbox}
 
  \begin{definition}[Star-Sampler Tree]\label{def:sst}
 	The star-sampler tree $\TT(v,\bw)$ for the tuple $(v,\bw)$ sampled by $\ssestimator(G,S_{\ell})$ is the following 2-level tree $\TT$ (with the same attributes as in Definition~\ref{def:cst}) 
	with only two nodes: 
	\begin{itemize}
		\item For the root $\alpha_r$ of $\TT$, $\Label{\alpha_r} := v$ and $\Value{\alpha_r} := 2m/d_v$. 
		
		\emph{(}$\Value{\alpha_r}$ is equal to the inverse
		of the probability that $v$ is sampled by $\ssestimator$\emph{)}. 
		
		\item The root $\alpha_r$ has exactly one child-node $\alpha_l$ in $\TT$ with $\Label{\alpha_l} = \bw = (w_1,\ldots,w_{\ell})$ and $\Value{\alpha_l} = {{d_v}\choose{\ell}}$.
		
		\emph{(}$\Value{\alpha_l}$ is equal to the inverse
		of the probability that $\bw$ is sampled by $\ssestimator$, conditioned on $v$ being sampled\emph{)}. 
	\end{itemize}
	Moreover, for the root-to-leaf path $\PP := (\alpha_r,\alpha_l)$, we define $\Label{\PP} := \Label{\alpha_r} \cup \Label{\alpha_l}$ and $\Value{\PP} := \Value{\alpha_r} \cdot \Value{\alpha_l}$. \emph{(}$\Label{\PP}$ is a representation of the star $S_{\ell}$ as defined in Section~\ref{sec:profiles}\emph{)}.
 \end{definition}
 \noindent
 See Figures~\ref{fig:B} and~\ref{fig:C} for an illustration of star-sampler trees. 
 
$\ssestimator$ can be implemented in our query model by using one edge-sample query in Line~(\ref{line2:sample-query}) and then picking one of the endpoints uniformly at random, 
a degree query to determine the degree of $v$, and $\ell$ neighbor queries in Line~(\ref{line2:neighbor-query}), resulting in $O(\ell)$ queries in total. 
It is also straightforward to verify that we can compute the star-sampler tree $\TT$ of an execution of $\ssestimator$ with no further queries and in $O(1)$ time.

We again define a process for estimating the number of stars in a graph using the information stored in the star-sampler tree and the $\ssestimator$ algorithm, as a warm-up to our main result in the next section. 

\paragraph{Warm-up: An Estimator for Stars.} The star-sampler tree $\TT$ is a random variable depending on the randomness of $\ssestimator$. We define the 
random variable $X$ such that $X := \Value{\PP}$ for the root-to-leaf path of $\TT$ iff $\Label{\PP}$ forms a copy of $S_{\ell}$ in $G$ and otherwise $X := 0$.  Our estimator algorithm can compute the value of this random variable
using only the information stored in the tree $\TT$ with no further queries to the graph (by simply checking if all $w_i$'s in $\bw$ are distinct). As such, the query
complexity and runtime of the estimator algorithm is still $O(1)$. We now prove,

\begin{lemma}\label{lem:ssestimator-output}
	For the random variable $X$ associated with $\ssestimator(G,S_\ell)$, 
	\begin{align*}
		\expect{X} &= (\sS), \qquad 
		\var{X} \leq 2m^{\ell} \cdot \expect{X}.
	\end{align*}
\end{lemma}
\begin{proof}
	Firstly, we have $X = \Value{\alpha_r} \cdot \Value{\alpha_l} \cdot \II\paren{\text{$\Label{\alpha_r}\cup \Label{\alpha_l}$ forms a copy of $S_{\ell}$}}$. As such, 
	\begin{align*}
	\expect{X} &=  \sum_{v \in V}\sum_{\substack{\bw  \in N(v)^{\ell}}}\Pr\paren{\Label{\alpha_r}=v} \cdot \Pr\paren{\Label{\alpha_l} = \bw} \\
		 &\hspace{4.5cm} \cdot \II(\text{$(v,\bw)$ forms a copy of $S_{\ell}$}) \cdot \Value{\alpha_r} \cdot \Value{\alpha_l}	\\
		&=	\sum_{v}\frac{d_v}{2m} \cdot \sum_{\bw} \frac{1}{{{d_v}\choose{\ell}}}  \cdot \II(\text{$(v,\bw)$ forms a copy of $S_{\ell}$}) \cdot (2m/d_v) \cdot {{d_v}\choose{\ell}} \\
		&= \sum_{v}\sum_{\bw} \II(\text{$(v,\bw)$ forms a copy of $S_{\ell}$}) = (\sS).
	\end{align*}
	This proves the desired bound on the exception. We now bound $\var{X}$.
	\begin{align*}
		\var{X} &\leq \expect{X^2} =  \sum_{v \in V}\sum_{\substack{\bw  \in N(v)^{\ell}}}\Pr\paren{\Label{\alpha_r}=v} \cdot \Pr\paren{\Label{\alpha_l} = \bw} \\
		 &\hspace{4.5cm} \cdot \II(\text{$(v,\bw)$ forms a copy of $S_{\ell}$}) \cdot \paren{\Value{\alpha_r} \cdot \Value{\alpha_l}}^2 \\
		&= \sum_{v }\frac{d_v}{2m} \cdot \sum_{\bw} \frac{1}{{{d_v}\choose{\ell}}}  \cdot \II(\text{$(v,\bw)$ forms a copy of $S_{\ell}$}) \cdot \paren{(2m/d_v) \cdot {{d_v}\choose{\ell}}}^2 \\
		&= \sum_{v } \sum_{\bw} \II(\text{$(v,\bw)$ forms a copy of $S_{\ell}$}) \cdot (2m/d_v) \cdot {{d_v}\choose{\ell}} \\
		&\leq \sum_{v } \sum_{\bw} \II(\text{$(v,\bw)$ forms a copy of $S_{\ell}$}) \cdot 2m \cdot {d_v}^{\ell-1} \tag{since ${{d_v}\choose{\ell}} \leq d_v^{\ell}$}\\
		&\leq 2m^{\ell} \cdot \sum_{v} \sum_{\bw} \II(\text{$(v,\bw)$ forms a copy of $S_{\ell}$}) \tag{since $d_v \leq m$} \\
		&= 2m^{\ell} \cdot (\sS) = 2m^{\ell} \cdot \expect{X},
	\end{align*}
	by the bound on $\expect{X}$ in the first part. 
\end{proof}

\begin{remark}\label{rem:star-variance}
	As we pointed out earlier, the bounds achieved by our Theorem~\ref{thm:main-intro} for counting stars are suboptimal in the light of the results in~\cite{AliakbarpourBGP18}. In Appendix~\ref{app:star-variance}, we show that 
	in fact our estimator in this section---using \emph{a different analysis} which is similar to that of~\cite{AliakbarpourBGP18}---also matches the optimal bounds achieved by~\cite{AliakbarpourBGP18}. 
	This suggests that even for the case of stars our algorithm in Theorem~\ref{thm:main-intro} is still optimal even though the general bounds in the theorem statement are not. 
	We note that our main estimator algorithm relies on the particular analysis of the estimator
	for stars presented in this section as the alternate analysis does not seem to compose with the rest of the argument.	
\end{remark}

\subsubsection*{The Estimator Algorithm for Arbitrary Subgraphs}

We now present our main estimator for the number of occurrences of an arbitrary subgraph $H$ in $G$, denoted by $(\sH)$. Recall the decomposition $\DD(H):=\set{\CC_1,\ldots,\CC_o,\SS_1,\ldots,\SS_s}$ of $H$ introduced in 
Section~\ref{sec:edge-cover}. Our algorithm creates a \emph{subgraph-sampler tree $\TT$} (a generalization of cycle-sampler and star-sampler trees in Definitions~\ref{def:cst} and~\ref{def:sst}) 
and use it to estimate $(\sH)$. We define the subgraph-sampler tree $\TT$ and the algorithm $\ssampler(G,H)$ that creates it simultaneously: 

\paragraph{Subgraph-Sampler Tree.} The subgraph-sampler tree $\TT$ is a ${z}$-level tree  for $z:=(2o+2s)$ returned by $\ssampler(G,H)$. The algorithm \ssampler constructs $\TT$ as follows (see Figure~\ref{fig:samplers} for an illustration).

\emph{Sampling Odd Cycles.} In $\ssampler(G,H)$, we first run $\scestimator(G,\CC_1)$ and initiate $\TT$ to be its output cycle-sampler tree. 
For every (current) leaf-node $\alpha$ of $\TT$, we run $\scestimator(G,\CC_2)$  \emph{independently} to obtain a cycle-sampler tree $\TT_{\alpha}$ (we say that $\alpha$ \emph{started} the sampling of $\TT_{\alpha}$). 
We then extend the tree $\TT$ with two new layers by connecting each leaf-node $\alpha$ to the root of $\TT_{\alpha}$ that started its sampling. This creates a $4$-level tree $\TT$. We continue like this for $o$ steps, each time 
appending the tree obtained by $\scestimator(G,\CC_j)$ for $j \in [o]$, to the (previous) leaf-node that started this sampling. This results in a $(2o)$-level tree. Note that the nodes in the tree $\TT$ can have different degrees as the number of leaf-nodes in 
the cycle-sampler tree is not necessarily the same always (not even for two different trees associated with one single $\CC_j$ through different calls to $\scestimator(G,\CC_j)$). 

\emph{Sampling Stars.} Once we iterated over all odd cycles of $\DD(H)$, we switch to processing stars $\SS_1,\ldots,\SS_s$. The approach is identical to the previous part. Let $\alpha$ be 
a (current) leaf-node of $\TT$. We run $\ssestimator(G,\SS_1)$ to obtain a star-sampler tree $\TT_{\alpha}$ and connect $\alpha$ to $\TT_{\alpha}$ to extend the levels of tree by $2$ more. We continue like
this for $s$ steps, each time appending the tree obtained by $\ssestimator(G,\SS_j)$ for $j \in [s]$, to the (former) leaf-node that started this sampling. This results in a $z$-level tree $\TT$. 
Note that all nodes added when sampling stars have exactly one child-node (except for the leaf-nodes) as by Definition~\ref{def:sst}, star-sampler trees always contain only two nodes. 

\emph{Labels and Values.} Each node $\alpha$ of $\TT$ is again given two attributes, $\Label{\alpha}$ and $\Value{\alpha}$, which are defined to be exactly the same attributes in the corresponding cycle-sampler 
or star-sampler tree that was used to define these nodes (recall that each node of $\TT$ is ``copied'' from a node in either a cycle-sampler or a star-sampler tree). Finally, for each root-to-leaf path $\PP$ in $\TT$, 
we define $\Label{\PP} := \bigcup_{\alpha \in \PP} \Label{\alpha}$  and $\Value{\PP} := \prod_{\alpha \in \PP} \Value{\alpha}$. In particular, 
$\Label{\PP} := ((\be_1,w_1),\ldots,(\be_o,w_o),(v_1,\bw_1),\ldots,(v_s,\bw_s))$ by definition of labels of cycle-sampler and star-sampler trees. As such $\Label{\PP}$ is a 
representation of the subgraph $H$ as defined in Section~\ref{sec:profiles}. By making $O(1)$ additional pair-queries 
to query all the remaining edges of this representation of $H$ we can determine whether $\Label{\PP}$ forms a copy of $H$ or not.

This concludes the description of $\ssampler(G,H)$ and its output subgraph-sampler tree $\TT$. We start analyzing this algorithm by bounding its query complexity. 

\begin{figure}[t!]
    \centering
\begin{tabular}[b]{cc}
    \begin{tabular}[b]{c}
      \begin{subfigure}[t]{0.4\columnwidth}
      \centering
           \begin{tikzpicture}[ auto ,node distance =1cm and 2cm , on grid , semithick , state/.style ={ circle ,top color =white , bottom color = white , draw, black , text=black}, every node/.style={inner sep=0,outer sep=0}]

\node[state, circle, black, line width=0.25mm, minimum height=35pt, minimum width=35pt] (r){$e_1,e_2$};
\node[state, circle, black, line width=0.25mm, minimum height=35pt, minimum width=35pt] (a1)[below left= 45pt and 45pt of r]{$w_1$};
\node[state, circle, black, line width=0.25mm, minimum height=35pt, minimum width=35pt] (a2)[below=45pt of r]{$w_2$};
\node[state, circle, black, line width=0.25mm, minimum height=35pt, minimum width=35pt] (a3)[below right= 45pt and 45pt of r]{$w_3$};

\node[inner sep=5pt, draw, dashed, ForestGreen, fit=(r) (a1) (a2) (a3)] {};

\draw (r) to (a1);
\draw (r) to (a2);
\draw (r) to (a3);
\end{tikzpicture}
        \caption{A cycle-sampler tree for $C_5$.}
        \label{fig:A}
      \end{subfigure} \\ 
      \begin{subfigure}[b]{0.4\columnwidth}
      \centering
     \begin{tikzpicture}[ auto ,node distance =1cm and 2cm , on grid , semithick , state/.style ={ circle ,top color =white , bottom color = white , draw, black , text=black}, every node/.style={inner sep=0,outer sep=0}]
\node[state, circle, black, line width=0.25mm, minimum height=35pt, minimum width=35pt] (b1){$v_{1}$};
\node[state, circle, black, line width=0.25mm, minimum height=35pt, minimum width=35pt] (b2)[below =45pt of b1]{$w_{1}$};

\node[inner sep=5pt, draw, dotted, red, fit=(b1) (b2), line width=0.25mm] {};

\draw (b1) to (b2);

\end{tikzpicture}
        \caption{A star-sampler tree for $S_1$.}
        \label{fig:B}
     \end{subfigure}      \\ 
          \begin{subfigure}[t]{0.4\columnwidth}
      \centering
           \begin{tikzpicture}[ auto ,node distance =1cm and 2cm , on grid , semithick , state/.style ={ circle ,top color =white , bottom color = white , draw, black , text=black}, every node/.style={inner sep=0,outer sep=0}]

\node[state, circle, black, line width=0.25mm, minimum height=35pt, minimum width=35pt] (b1){$v_{1}$};
\node[state, circle, black, line width=0.25mm, minimum height=35pt, minimum width=35pt] (b2)[below =45pt of b1]{$w_{1},w_{2}$};

\node[inner sep=5pt, draw, dotted, red, fit=(b1) (b2), line width=0.25mm] {};

\draw (b1) to (b2);
\end{tikzpicture}
        \caption{A star-sampler tree for $S_2$.}
        \label{fig:C}
        \end{subfigure}
    \end{tabular}
    &
    \begin{subfigure}[b]{0.4\columnwidth}
    \centering
  \begin{tikzpicture}[ auto ,node distance =1cm and 2cm , on grid , semithick , state/.style ={ circle ,top color =white , bottom color = white , draw, black , text=black}, every node/.style={inner sep=0,outer sep=0}]

\node[state, circle, black, line width=0.25mm, minimum height=35pt, minimum width=35pt] (r){$e_1,e_2$};
\node[state, circle, black, line width=0.25mm, minimum height=35pt, minimum width=35pt] (a1)[below left= 45pt and 45pt of r]{};
\node[state, circle, black, line width=0.25mm, minimum height=35pt, minimum width=35pt] (a2)[below=45pt of r]{};
\node[state, circle, black, line width=0.25mm, minimum height=35pt, minimum width=35pt] (a3)[below right= 45pt and 45pt of r]{$w_1$};

\node[inner sep=5pt, draw, dashed, ForestGreen, fit=(r) (a1) (a2) (a3)] {};

\draw (r) to (a1);
\draw (r) to (a2);
\draw[blue, line width=0.5mm]  (r) to (a3);

\node[state, circle, black, line width=0.25mm, minimum height=35pt, minimum width=35pt] (b1)[below =45pt of a1]{};
\node[state, circle, black, line width=0.25mm, minimum height=35pt, minimum width=35pt] (b2)[below =45pt of b1]{};

\node[inner sep=5pt, draw, dotted, red, fit=(b1) (b2), line width=0.25mm] {};

\node[state, circle, black, line width=0.25mm, minimum height=35pt, minimum width=35pt] (c1)[below = 45pt of b2]{};
\node[state, circle, black, line width=0.25mm, minimum height=35pt, minimum width=35pt] (c2)[below = 45pt of c1]{};

\node[inner sep=5pt, draw, dotted, red, fit=(c1) (c2), line width=0.25mm] {};

\draw (a1) to (b1);
\draw (b1) to (b2);
\draw (b2) to (c1);
\draw (c1) to (c2);

\node[state, circle, black, line width=0.25mm, minimum height=35pt, minimum width=35pt] (b11)[below =45pt of a2]{};
\node[state, circle, black, line width=0.25mm, minimum height=35pt, minimum width=35pt] (b12)[below =45pt of b11]{};

\node[inner sep=5pt, draw, dotted, red, fit=(b11) (b12), line width=0.25mm] {};

\node[state, circle, black, line width=0.25mm, minimum height=35pt, minimum width=35pt] (c11)[below = 45pt of b12]{};
\node[state, circle, black, line width=0.25mm, minimum height=35pt, minimum width=35pt] (c12)[below = 45pt of c11]{};

\node[inner sep=5pt, draw, dotted, red, fit=(c11) (c12), line width=0.25mm] {};

\draw (a2) to (b11);
\draw (b11) to (b12);
\draw (b12) to (c11);
\draw (c11) to (c12);

\node[state, circle, black, line width=0.25mm, minimum height=35pt, minimum width=35pt] (b21)[below =45pt of a3]{$v_{1}$};
\node[state, circle, black, line width=0.25mm, minimum height=35pt, minimum width=35pt] (b22)[below =45pt of b21]{$w_{2}$};

\node[inner sep=5pt, draw, dotted, red, fit=(b21) (b22), line width=0.25mm] {};

\node[state, circle, black, line width=0.25mm, minimum height=35pt, minimum width=35pt] (c21)[below = 45pt of b22]{$v_{2}$};
\node[state, circle, black, line width=0.25mm, minimum height=35pt, minimum width=35pt] (c22)[below = 45pt of c21]{$w_{3},w_{4}$};

\node[inner sep=5pt, draw, dotted, red, fit=(c21) (c22), line width=0.25mm] {};

\draw[blue, line width=0.5mm] (a3) to (b21);
\draw[blue, line width=0.5mm]  (b21) to (b22);
\draw[blue, line width=0.5mm]  (b22) to (c21);
\draw[blue, line width=0.5mm]  (c21) to (c22);
\end{tikzpicture}

      \caption{A subgraph-sampler tree for the subgraph $H$ which is decomposed in $\DD(H)$ to a $C_3$ cycle (length $3$), and two stars $S_1$ (one petal) and $S_2$ (two petals). The labels of some nodes are omitted in the figure.}
      \label{fig:D}
    \end{subfigure}
  \end{tabular}
   \caption{Illustration of the sampler-subgraph $\TT$ for the subgraph $H$ of Figure~\ref{fig:decomposition}. The (blue) thick line in part~(d) shows a root-to-leaf path $\PP$ with 
 $
  \Value{\PP} = (m^{2}) \cdot (\dstar_{(e_1,e_2)}) \cdot (\frac{2m}{d_{v_1}}) \cdot {{d_{v_1}}\choose{1}} \cdot (\frac{2m}{d_{v_2}}) \cdot {{d_{v_2}}\choose{2}}. 
$
The variable $X$ of $\PP$ is equal to $\Value{\PP}$ iff the profile $((e_1,e_2,w_1),(v_1,w_2),(v_2,w_2,w_3))$ forms a copy of $H$ in $G$. 
  }
 \label{fig:samplers}

\end{figure}
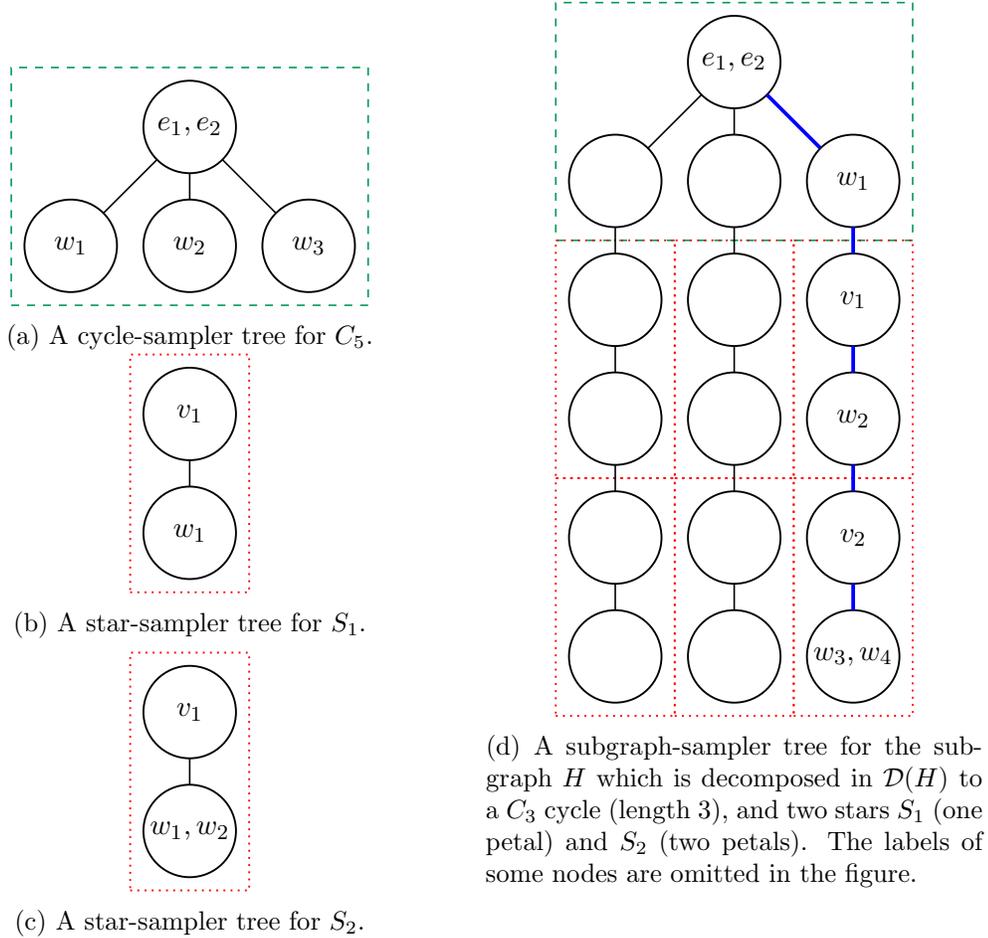

\begin{lemma}\label{lem:sampler-time}
	The expected query complexity and running time of $\ssampler$ is $O(1)$. 
\end{lemma}
\begin{proof}
	As was the case for $\scestimator$ and $\ssestimator$, in the $\ssampler$ also the query complexity of the algorithm is within a constant factor of number of nodes in the subgraph-sampler tree $\TT$ that it returns. Hence, we 
	only need to bound the number of nodes in $\TT$. 
	
	Let $\LL_1,\ldots,\LL_z$ denote the set of nodes in layers $1$ to $z$ of $\TT$. $\LL_1$ contains only the root $\alpha_r$ of $\TT$. Let $\be_r:=\Label{\alpha_r}$. 
	By definition of the cycle-sampler tree that forms the first two layers of $\TT$, the number of child-nodes of $\alpha_r$ is $t_{\be_r}$ (defined in Line~(\ref{line1:for-loop}) of $\scestimator$). 
	As such ${\card{\LL_2}} = t_{\be_1}$. The nodes in any even layer of $\TT$ have only a single child-node by construction, hence $\card{\LL_3} = \card{\LL_2}$. Now, 
	for each node $\alpha$ in $\LL_3$ with label $\be_{\alpha} := \Label{\alpha}$, the number of child-nodes is $t_{\be_\alpha}$ (the number of child-nodes
	of a node in an even layer is always one by construction). Hence, $\card{\LL_4} = \sum_{\alpha \in \LL_3} t_{\be_{\alpha}}$. By continuing like this, we obtain that for each layer $\LL_{2i}$ for $i \leq o$, 
	$\card{\LL_{2i}} = \sum_{\alpha \in \LL_{2i-1}} t_{\be_{\alpha}}$. After this, we reach the layers corresponding to star-sampler subgraphs, in which every non-leaf node has exactly one child-node, and hence 
	$\card{\LL_{j}} = \card{\LL_{2o}}$ for all $j \geq 2o$. Moreover note that each odd-layer node of $\TT$ \emph{independently} starts the sampling of its subtree (even independently of its parent-nodes). 
	
	Define $t_1,\ldots,t_o$ as the random variables denoting the number of leaf-nodes in the cycle-sampler trees for $\CC_1,\ldots \CC_o$ in the decomposition $\DD(H)$. 
	By Lemma~\ref{lem:scestimator-time}, $\expect{t_i} = O(1)$ for all $i \in [o]$. 
	By the above discussion, we have, 
	\begin{align*}
		\expect{\card{\LL_1 \cup \ldots \cup \LL_z}} &= O(z) \cdot \expect{t_1} \cdot \expect{t_2} \cdot \ldots \cdot \expect{t_o} + O(z) \Eq{Lemma~\ref{lem:scestimator-time}} O(1)^{O(z)} = O(1),
	\end{align*}
	as $z = O(1)$ since size of $H$ is constant. Finally, note that running time of the algorithm is at most a constant factor larger than its query complexity. 
\end{proof}

We are now ready to present our estimator algorithm using $\ssampler$ and the subgraph-sampler tree $\TT$ it outputs. 

\paragraph{An Estimator for Arbitrary Subgraphs.} Note that as before the subgraph-sampler tree $\TT$ itself is a random variable depending on the randomness of $\ssampler$. 
For any root-to-leaf path $\PP_i := \alpha_1,\ldots,\alpha_z$ of $\TT$, we define the random variable $X_{i}$ such that $X_{i} := \Value{\PP_i}$ iff $\Label{\PP_i}$ forms a copy of $H$ in $G$ and 
otherwise $X_{i} :=0$.  We further define $Y := (\frac{1}{t} \sum_{i=1}^{t} X_i)$, where $t$ is the number of leaf-nodes of $\TT$ (which itself is a random variable). These random variables can all be computed
from $\TT$ and $\ssampler$ with at most $O(1)$ further pair-queries per each root-to-leaf path $\PP$ of the tree to determine if indeed $\Label{\PP}$ forms a copy of $H$ in $G$ or not. As such, query complexity and runtime of 
this algorithm is proportional to $\ssampler$ (which in expectation is $O(1)$ by Lemma~\ref{lem:sampler-time}). In the following two lemmas, we show that $Y$ is a low-variance unbiased estimator of $(\sH)$. To continue, we first need some notation.

\paragraph{Notation.} For any node $\alpha$ in $\TT$, we use $\TT_{\alpha}$ to denote the sub-tree of $\TT$ rooted at $\alpha$. For a leaf-node $\alpha$, we define a random variable $Y_{\alpha}$ which is $\Value{\alpha}$ iff
for the root-to-leaf path $\PP$ ending in $\alpha$, $\Label{\PP}$ forms a copy of $H$ in $G$ and otherwise $Y_{\alpha}$ is $0$. For an internal node $\alpha$ in $\TT$ with $t$ child-nodes $\alpha_1,\ldots,\alpha_t$, 
we define $Y_{\alpha} = \Value{\alpha} \cdot \paren{\frac{1}{t} \cdot \sum_{i=1}^{t} Y_i}$. It is easy to verify that $Y_{\alpha_r}$ for the root $\alpha_r$ of $\TT$ is the same as the estimator random variable $Y$ defined earlier. 
Furthermore, for a node $\alpha$ in level $\ell$ of $\TT$, we define $\bL{\alpha} := (\Label{\alpha_1},\Label{\alpha_2},\ldots,\Label{\alpha_{\ell-1}})$, where $\alpha_1,\ldots,\alpha_{\ell-1}$ forms the path from the root of $\TT$ 
to the parent of $\alpha$. 

We analyze the expected value and the variance of the estimator in the following two lemmas. 
\begin{lemma}\label{lem:sampler-expect}
	For the random variable $Y$ for $\ssampler(G,H)$, $\expect{Y} = (\sH)$. 
\end{lemma}
\begin{proof}
	We prove this lemma inductively by showing that for any node $\alpha$ in an \emph{odd layer} of $\TT$,
	\begin{align*}
		\expect{Y_{\alpha} \mid \bL{\alpha}} = \ssH{\bL{\alpha}},
	\end{align*} 
	where $\ssH{\bL{\alpha}}$ denotes the number of copies of $H$ in $G$ that contain  the vertices and edges specified by $\bL{\alpha}$ (according to the decomposition $\DD(H)$). $\expect{Y_{\alpha} \mid \bL{\alpha}}$
	measures the value of $Y_{\alpha}$ after we fix the rest of the tree $\TT$ and let the sub-tree $\TT_{\alpha}$ be chosen randomly as in $\ssampler$. 
	
	The base case of the induction, i.e., for vertices in the last odd layer of $\TT$ follows exactly as in the proofs of Lemmas~\ref{lem:scestimator-output} and~\ref{lem:ssestimator-output} (as will also become evident shortly) and hence
	we do not repeat it here. We now prove the induction hypothesis. Fix a vertex $\alpha$ in an odd layer $\ell$. We consider two cases based on
	whether $\ell < 2o$ (hence $\alpha$ is root of a cycle-sampler tree) or $\ell > 2o$ (hence $\alpha$ is root of a star-sampler tree).
	
	\noindent
	\emph{Case of $\ell < 2o$}. In this case, the sub-tree $\TT_{\alpha}$ in the next two levels is a cycle-sampler tree, hence,  
	\begin{align*}
		\expect{Y_{\alpha} \mid \bL{\alpha}} &= \sum_{\be} \Pr\paren{\Label{\alpha} = \be} \cdot \Value{\alpha} \cdot \paren{\frac{1}{t_{\be}} \sum_{i=1}^{t_{\be}} \expect{Y_{\alpha_i} \mid \bL{\alpha},\be}} \tag{here, $\alpha_i$'s are 
		child-nodes of $\alpha$}\\
		&= \sum_{\be} {\frac{1}{t_{\be}} \sum_{i=1}^{t_{\be}} \expect{Y_{\alpha_i} \mid \bL{\alpha},\be}} \tag{as by definition, $\Value{\alpha} = \Pr\paren{\Label{\alpha} = \be}^{-1}$}
	\end{align*}
	Note that each $\alpha_i$ has exactly one child-node, denoted by $\beta_i$. As such, 
	\begin{align*}
		\expect{Y_{\alpha} \mid \bL{\alpha}} &= \sum_{\be} {\frac{1}{t_{\be}} \sum_{i=1}^{t_{\be}} \expect{Y_{\alpha_i} \mid \bL{\alpha},\be}}  \\
		&= \sum_{\be} {\frac{1}{t_{\be}} \sum_{i=1}^{t_{\be}} \sum_{w} \Pr\paren{\Label{\alpha_i}=w} \cdot \Value{\alpha_i} \cdot \expect{Y_{\beta_i} \mid \bL{\alpha},\be,w}} \\
		&= \sum_{\be}\frac{1}{t_{\be}} \sum_{i=1}^{t_{\be}} \sum_{w} \expect{Y_{\beta_i} \mid \bL{\beta_i}} \tag{by definition $\Value{\alpha_i} = \Pr\paren{\Label{\alpha_i} = w}^{-1}$ and $\bL{\beta_i} = \bL{\alpha},(\be,w)$} \\
		&= \sum_{\be}\frac{1}{t_{\be}} \sum_{i=1}^{t_{\be}} \sum_{w} \ssH{\bL{\beta_i}} = \sum_{\be}\frac{1}{t_{\be}} \sum_{i=1}^{t_{\be}} \sum_{w} \ssH{\bL{\alpha},(\be,w)} \tag{by induction hypothesis for odd-layer nodes $\beta_i$'s} \\
		&=\sum_{\be}\sum_{w} \ssH{\bL{\alpha},(\be,w)} = \ssH{\bL{\alpha}}.
	\end{align*}
	This concludes the proof of induction hypothesis in this case. Note that this proof was basically the same proof for the expectation bound of the estimator for cycle-sampler tree in Lemma~\ref{lem:scestimator-output}.  
	
	\emph{Case of $\ell > 2o$}. In this case, the sub-tree $\TT_{\alpha}$ in the next two levels is a star-sampler tree. By the same analogy made in the proof of the previous part and Lemma~\ref{lem:scestimator-output}, 
	the proof of this part also follows directly from the proof of Lemma~\ref{lem:ssestimator-output} for star-sampler trees. We hence omit the details. 
	
	We can now finalize the proof of Lemma~\ref{lem:sampler-expect}, by noting that for the root $\alpha_r$ of $\TT$, $\bL{\alpha_r}$ is the empty-set and hence, $\expect{Y} = \expect{Y_{\alpha_r} \mid \bL{\alpha_r}}$, 
	which by induction is equal to $(\sH)$. 
\end{proof}
\noindent
Recall that $\rho(H)$ is the fractional edge-cover number of $H$ and it is related to $\DD(H)$ through Eq~(\ref{eq:agm}). 

\begin{lemma}\label{lem:sampler-variance}
	For the random variable $Y$ for $\ssampler(G,H)$, $\var{Y} = O(m^{\rho(H)}) \cdot \expect{Y}$. 
\end{lemma}
\begin{proof}
	We bound $\var{Y}$ using a similar inductive proof as in Lemma~\ref{lem:sampler-expect}. Recall the parameters $\rhoC_1,\ldots,\rhoC_o$ and $\rhoS_1,\ldots,\rhoS_s$ associated respectively with the cycles $\CC_1,\ldots,\CC_o$
	and stars $\SS_1,\ldots,\SS_s$ of the decomposition $\DD(H)$. For simplicity of notation, for any $i \in [o+s]$, we define $\rhoG{i}$ as follows:
	\begin{align*}
		&\textnormal{for all $i \leq o$, $\rhoG{i} := \sum_{j=i}^{o} \rhoC_j + \sum_{j=1}^{s} \rhoS_j$,}\qquad 
		&\textnormal{for all $o < i \leq o+s$, $\rhoG{i} := \sum_{j=i-o}^{s} \rhoS_j$.}
	\end{align*}
	We inductively show that, 
	for any node $\alpha$ in an \emph{odd layer} $2\ell-1$ of $\TT$,
	\begin{align*}
		\var{Y_{\alpha} \mid \bL{\alpha}} \leq 2^{2z-2\ell} \cdot m^{\rhoG{\ell}} \cdot \ssH{\bL{\alpha}},
	\end{align*} 
	where $\ssH{\bL{\alpha}}$ denotes the number of copies of $H$ in $G$ that contain the  vertices and edges specified by $\bL{\alpha}$ (according to the decomposition $\DD(H)$). 
	
	The induction is from the leaf-nodes of the tree to the root. 
	The base case of the induction, i.e., for vertices in the last odd layer of $\TT$ follows exactly as in the proofs of Lemmas~\ref{lem:scestimator-output} and~\ref{lem:ssestimator-output} (as will also become evident shortly) and hence
	we do not repeat it here. We now prove the induction hypothesis. Fix a vertex $\alpha$ in an odd layer $2\ell-1$. We consider two cases based on
	whether $\ell \leq o$ (hence $\alpha$ is root of a cycle-sampler tree) or $\ell > o$ (hence $\alpha$ is root of a star-sampler tree).
	
	\noindent
	\emph{Case of $\ell \leq o$}. In this case, the sub-tree $\TT_{\alpha}$ in the next two levels is a cycle-sampler tree corresponding to the odd cycle $\CC_\ell$ of $\DD(H)$. Let the number of edges in $C_{\ell}$ be $(2k+1)$ (i.e., $\CC_{\ell} = C_{2k+1}$) 
	Let $\be$ denote the label of the $\alpha$. 
	By the  law of total variance in Eq~(\ref{eq:total-variance}) 
	\begin{align}
		\var{Y_{\alpha} \mid \bL{\alpha}} &= \expect{\var{Y_{\alpha} \mid \be} \mid \bL{\alpha}} + \var{\expect{Y_{\alpha} \mid \be} \mid  \bL{\alpha}}. \label{eq:svar1}
	\end{align}
	We start by bounding the second term in Eq~(\ref{eq:svar1}) which is easier. By the inductive proof of Lemma~\ref{lem:sampler-expect}, we also have, 
	$\expect{Y_{\alpha} \mid \bL{\alpha},\be} = \ssH{\bL{\alpha},\be}$. As such, 
	\begin{align}
		\var{\expect{Y_{\alpha} \mid \be} \mid  \bL{\alpha}} &= \var{\ssH{\bL{\alpha},\be} \mid \bL{\alpha}} \leq \expect{\ssH{\bL{\alpha},\be}^2 \mid \bL{\alpha}} \notag \\
		&= \sum_{\be} \Pr\paren{\Label{\alpha} = \be} \cdot \ssH{\bL{\alpha},\be}^2 =  \frac{1}{m^{{k}}} \sum_{\be} \ssH{\bL{\alpha},\be}^2 \tag{$\Pr\paren{\Label{\alpha}=\be}=1/m^{k}$ by definition of $\scestimator$} \\
		&\leq \frac{1}{m^{{k}}} \Paren{\sum_{\be} \ssH{\bL{\alpha},\be}}^2 = \frac{1}{m^{{k}}} {\ssH{\bL{\alpha}}}^2 \notag \\
		&\leq m^{\rhoG{\ell}} \cdot \ssH{\bL{\alpha}}. \label{eq:svar2}
	\end{align}
	The reason behind the last equality is that $\ssH{\bL{\alpha}}$ is at most equal to the number of copies of the subgraph of $H$ consisting of $\CC_{\ell},\ldots,\CC_{o},\SS_1,\ldots,\SS_{s}$, which by 
	Lemma~\ref{lem:subgraph-AGM} is at most $m^{\rhoG{\ell}}$ by definition of $\rhoG{\ell}$. We now bound the first and the main term in Eq~(\ref{eq:svar1}),
	\begin{align*}
		\expect{\var{Y_{\alpha} \mid \be} \mid \bL{\alpha}} &= \sum_{\be}\Pr\paren{\Label{\alpha}=\be} \cdot \var{Y_{\alpha} \mid \be,\bL{\alpha}} \\
		&= \sum_{\be} \frac{1}{m^k} \cdot m^{2k} \cdot \frac{1}{t_{\be}^2} \cdot \sum_{i=1}^{t_{\be}} \var{Y_{\alpha_i} \mid \be,\bL{\alpha}}, \tag{here $\alpha_i$'s are child-nodes of $\alpha$} 
	\end{align*}
	where the final equality holds because $Y_{\alpha_i}$'s are independent conditioned on $\be,\bL{\alpha}$ and since  $Y_{\alpha}$ is by definition $m^{k}$ 
	times the average of $Y_{\alpha_i}$'s. Moreover, note that distribution of all $Y_{\alpha_i}$'s are the same. Hence, by canceling the terms, 
	\begin{align}
		\expect{\var{Y_{\alpha} \mid \be} \mid \bL{\alpha}} &= m^{k} \cdot  \sum_{\be} \frac{1}{t_{\be}} \cdot  \var{Y_{\alpha_1} \mid \be,\bL{\alpha}}, \label{eq:svar3}
	\end{align}
	We thus only need to bound $\var{Y_{\alpha_1} \mid \be,\bL{\alpha}}$. Recall that $\alpha_1$ corresponds to a leaf-node in a cycle-sampler tree and hence its label is a vertex $w$ from the neighborhood of 
	$\ustar_{\be}$ as defined in $\scestimator$. We again use the  law of total variance in Eq~(\ref{eq:total-variance}) to obtain,
	\begin{align}
		\var{Y_{\alpha_1} \mid \be,\bL{\alpha}} = \expect{\var{Y_{\alpha_1} \mid w} \mid \be,\bL{\alpha}} + \var{\expect{Y_{\alpha_1} \mid w} \mid \be,\bL{\alpha}} \label{eq:svar4}
	\end{align}
	For the first term, 
	\begin{align*}
		\expect{\var{Y_{\alpha_1} \mid w} \mid \be,\bL{\alpha}} &= \sum_{w \in N(\ustar_{\be})} \Pr\paren{\Label{\alpha_1} = w} \cdot \var{Y_{\alpha_1} \mid w ,\be,\bL{\alpha}} \\
		&=  \sum_{w} \frac{1}{\dstar_{\be}} \cdot \paren{\dstar_{\be}}^2 \cdot \var{Y_{\beta_1} \mid w ,\be,\bL{\alpha}},
	\end{align*}
	where $\beta_1$ is the unique child-node of $\alpha_1$ and so $Y_{\alpha_1} = \Value{\alpha_1} \cdot Y_{\beta_1}$, while conditioned on $\be$, $\Value{\alpha_1} = \dstar_{\be}$. 
	Moreover, as $\bL{\beta_1} = (\bL{\alpha},\be,w)$, and by canceling the terms,
	\begin{align}
		\expect{\var{Y_{\alpha_1} \mid w} \mid \be,\bL{\alpha}} &= \sum_{w} \dstar_{\be} \cdot \var{Y_{\beta_1} \mid \bL{\beta_1}} \notag \\
		&\leq \sum_{w}  \dstar_{\be} \cdot 2^{2z-2\ell-2} \cdot m^{\rhoG{(\ell+1)}} \cdot \ssH{\bL{\beta_1}}, \label{eq:svar5}
	\end{align}
	where the inequality is by induction hypothesis for the odd-level node $\beta_1$. We now bound the second term in Eq~(\ref{eq:svar4}) as follows,
	\begin{align}
		\var{\expect{Y_{\alpha_1} \mid w} \mid \be,\bL{\alpha}} &\leq \expect{\Paren{\expect{Y_{\alpha_1} \mid w}}^2 \mid \be,\bL{\alpha}} \notag \\
		&= \sum_{w} \Pr\paren{\Label{\alpha_1} = w} \cdot \Paren{\expect{Y_{\alpha_1} \mid w,\be,\bL{\alpha}}}^2 \notag \\
		&= \sum_{w} \frac{1}{\dstar_{\be}} \cdot \paren{\dstar_{\be}}^2 \cdot \Paren{\expect{Y_{\beta_1} \mid w,\be,\bL{\alpha}}}^2 \notag \\
		&= \sum_{w} \dstar_{\be} \cdot \Paren{\expect{Y_{\beta_1} \mid \bL{\beta_1}}}^2 
		= \sum_{w} \dstar_{\be} \cdot \ssH{\bL{\beta_1}}^2 \notag \\
		&\leq \sum_{w} \dstar_{\be} \cdot m^{\rhoG{(\ell+1)}} \cdot \ssH{\bL{\beta_1}}.\label{eq:svar6}
	\end{align}
	Here, the second to last equality holds by the inductive proof of Lemma~\ref{lem:sampler-expect}, and the last equality is because $\ssH{\bL{\beta_1}} \leq m^{\rhoG{(\ell+1)}}$ by Lemma~\ref{lem:subgraph-AGM}, 
	as $\ssH{\bL{\beta_1}}$ is at most equal to the total number of copies of a subgraph of $H$ on $\CC_{\ell+1},\ldots,\CC_{o},\SS_1,\ldots,\SS_s$ (and by definition of $\rhoG{(\ell+1)}$). 
	We now plug in Eq~(\ref{eq:svar5}) and Eq~(\ref{eq:svar6}) in Eq~(\ref{eq:svar4}), 
	\begin{align*}
		\var{Y_{\alpha_1} \mid \be,\bL{\alpha}} \leq \sum_{w} \dstar_{\be} \cdot \paren{2^{2z-2\ell-2} \cdot m^{\rhoG{(\ell+1)}} \cdot \ssH{\bL{\beta_1}}  + m^{\rhoG{(\ell+1)}} \cdot \ssH{\bL{\beta_1}}}.
	\end{align*}
	We now in turn plug this in Eq~(\ref{eq:svar3}), 
	\begin{align*}
		\expect{\var{Y_{\alpha} \mid \be} \mid \bL{\alpha}} &\leq m^{k}\sum_{\be} \frac{1}{t_{\be}}  \sum_{w} \dstar_{\be} \cdot  \paren{2^{2z-2\ell-2} \cdot m^{\rhoG{(\ell+1)}} \cdot \ssH{\bL{\beta_1}} + 
		m^{\rhoG{(\ell+1)}} \cdot \ssH{\bL{\beta_1}}} \\
		&\leq m^{k}\sqrt{m} \cdot \sum_{\be} \sum_{w} {2^{2z-2\ell-1} \cdot m^{\rhoG{(\ell+1)}} \cdot \ssH{\bL{\beta_1}}} \tag{as $t_{\be} \geq \dstar_{\be}/\sqrt{m}$} \\
		&\leq {2^{2z-2\ell-1} \cdot m^{\rhoG{\ell}} \cdot \sum_{\be} \sum_{w}  \ssH{\bL{\beta_1}}} \tag{as $\rhoC_\ell = k+1/2$ and $\rhoG{\ell} = \rhoC_{\ell} + \rhoG{(\ell+1)}$ by definition} \\
		&= 2^{2z-2\ell-1} \cdot m^{\rhoG{\ell}} \cdot  \ssH{\bL{\alpha}}  \tag{as $\bL{\beta_1} = (\bL{\alpha}, \be,w)$}.
	\end{align*}
	Finally, by plugging in this and Eq~(\ref{eq:svar2}) in Eq~(\ref{eq:svar1}), 
	\begin{align*}
		\var{Y_{\alpha} \mid \bL{\alpha}} &=  2^{2z-2\ell-1} \cdot m^{\rhoG{\ell}} \cdot  \ssH{\bL{\alpha}}  + m^{\rhoG{\ell}} \cdot  \ssH{\bL{\alpha}} \leq 2^{2z-2\ell} \cdot m^{\rhoG{\ell}} \cdot \ssH{\bL{\alpha}},
	\end{align*}
	finalizing the proof of induction step in this case. We again remark that this proof closely followed the proof for the variance of the estimator for cycle-sampler tree in Lemma~\ref{lem:scestimator-output}.  
	
	\emph{Case of $\ell > o$}. In this case, the sub-tree $\TT_{\alpha}$ in the next two levels is a star-sampler tree. By the same analogy made in the proof of the previous case and Lemma~\ref{lem:scestimator-output}, 
	the proof of this part also follows the proof of Lemma~\ref{lem:ssestimator-output} for star-sampler trees. We hence omit the details. 
	
	To conclude, we have that $\var{Y} = \var{Y_{\alpha_r} \mid \bL{\alpha_r}} = O(m^{\rho(H)}) \cdot (\sH) = O(m^{\rho(H)}) \cdot \expect{Y}$ as $Y = Y_{\alpha_r}$ for the root $\alpha_r$ of $\TT$, 
	$\bL{\alpha_r} = \emptyset$, $(\sH) = \expect{Y}$ by Lemma~\ref{lem:sampler-expect}, and $z = O(1)$. 
\end{proof}
 
 \newcommand{\est}{\ensuremath{\textnormal{\sc est}}}
 
 \subsection{An Algorithm for Estimating Occurrences of Arbitrary Subgraphs}\label{sec:final-alg}
 
 We now use our estimator algorithm from the previous section to design our algorithm for estimating the occurrences of an arbitrary subgraph $H$ in $G$. In the following theorem, we assume that the algorithm
 has knowledge of $m$ and also a lower bound on the value of $\sH$; these assumptions can be lifted easily as we describe afterwards. 
 
 \begin{theorem}\label{thm:alg-estimate}
There exists a sublinear time algorithm that uses degree, neighbor, pair, and edge sample queries and given a precision parameter $\eps \in (0,1)$, an explicit access to a constant-size graph $H(V_H,E_H)$, a query access to the input graph $G(V,E)$, 
the number of edges $m$ in $G$, and a lower bound $h \leq \sH$, with high probability outputs a $(1 \pm \eps)$-approximation to $\sH$ using: 
	\begin{align*}
		O\Paren{\min\set{m,\frac{m^{\rho(H)}}{h} \cdot \frac{\log{n}}{\eps^2}}} \textnormal{ queries and } O\Paren{\frac{m^{\rho(H)}}{h} \cdot \frac{\log{n}}{\eps^2}} \textnormal{ time,} 
	\end{align*}
	in the worst-case. 
\end{theorem}
\begin{proof}
	Fix a sufficiently large constant $c > 0$. We run $\ssampler(G,H)$ for $k := \frac{c \cdot m^{\rho(H)}}{\eps^2 \cdot h}$ time independently in parallel to obtain estimates $Y_1,\ldots,Y_k$ and let $Z:= \frac{1}{k} \sum_{i=1}^{k} Y_i$. 
	By Lemma~\ref{lem:sampler-expect}, $\expect{Z} = (\sH)$. Since $Y_i$'s are independent, we also have 
	\begin{align*}
	\var{Z} = \frac{1}{k^2} \sum_{i=1}^{k} \var{Y_i} \leq \frac{1}{k} \cdot O(m^{\rho(H)}) \cdot \expect{Z} \leq \frac{\eps^2}{10} \cdot \expect{Z}^2,
	\end{align*}
	by Lemma~\ref{lem:sampler-variance}, and by choosing the constant $c$ sufficiently larger than the constant in the O-notation of this lemma, together with the fact that $h \leq (\sH) = \expect{Z}$. By Chebyshev's inequality 
	(Proposition~\ref{prop:cheb}), 
	\begin{align*}
		\Pr\paren{\card{Z-\expect{Z}} \geq \eps \cdot \expect{Z}} \leq \frac{\var{Z}}{\eps^2 \cdot \expect{Z}^2} \leq \frac{1}{10},
	\end{align*}
	by the bound above on the variance. This means that with probability $0.9$, this algorithm outputs a $(1\pm \eps)$-approximation of $\sH$. Moreover, the expected query complexity and running time of this algorithm is 
	$O(k)$ by Lemma~\ref{lem:sampler-time}, which is $O(\frac{m^{\rho(H)}}{\eps^2})$ (if $k \geq m$, we simply query all edges of the graph and solve the problem using an offline enumeration algorithm). 
	To extend this result to a high probability bound and also making the guarantee of query complexity and run-time in the worst-case, 
	we simply run this algorithm $O(\log{n})$ times in parallel and stop each execution that uses more than $10$ times queries than the expected query complexity of the above algorithm.
\end{proof}

The algorithm in Theorem~\ref{thm:alg-estimate} assumes the knowledge of $h$ which is a lower bound on $(\sH)$. However, this assumption can be easily removed by making a geometric search 
on $h$ starting from $m^{\rho(H)}/2$ which is (approximately) the largest value for $(\sH)$ all the way down to $1$ in factors of $2$, and stopping the search once the estimates returned for a guess of $h$ became consistent with $h$ itself.  
This only increases the query complexity and runtime of the algorithm by $\polylog{(n)}$ factors. As this part is quite standard, we omit the details and instead refer the interested reader to~\cite{EdenLRS15,EdenRS18}. 
This concludes the proof of our main result in Theorem~\ref{thm:main-intro} from the introduction. 
 
 \subsection{Extension to the Database Join Size Estimation Problem}\label{sec:extensions}
 
As pointed out earlier in the paper, the database join size estimation for binary relations can be modeled by the subgraph estimation problem where the subgraph $H$ and the underlying graph $G$ are additionally \emph{edge-colored} and we are only interested in counting the copies of $H$ in $G$ with matching colors on the edges. In this abstraction, the edges of the graph $G$ correspond to the entries of the database, and the color of edges determine the relation of the entry. 
 
 We formalize this variant of the subgraph counting problem in the following. In the \emph{colorful} subgraph estimation problem, we are given a subgraph $H(V_H,E_H)$ with a coloring function $c_H: E_H \rightarrow \IN$ 
and query access to a graph $G(V,E)$ along with a coloring function $c_G: E \rightarrow \IN$. The set of allowed queries to $G$ contains the degree queries, pair queries, neighbor queries, and edge-sample queries as before, with a simple change that
whenever we query an edge (through the last three types of queries), the color of the edge according to $c_G$ is also revealed to the algorithm. 
Our goal is to estimate the number of copies of $H$ in $G$ with matching colors, i.e., the \emph{colorful} copies of $H$. 

It is immediate to verify that our algorithm in this section can be directly applied to the colorful subgraph estimation problem with the only difference that when testing whether a subgraph forms a copy of $H$ in $G$, we in fact check whether 
this subgraph forms a colorful copy of $H$ in $G$ instead. The analysis of this new algorithm is exactly as in the case of the original algorithm with the only difference that we switch the parameter $\sH$ to $\sH_c$ that only 
counts the number of copies of $H$ with the same colors in $G$. To summarize, we obtain an algorithm with $\Os(\frac{m^{\rho(H)}}{\sH_c})$ query and time complexity for the colorful subgraph counting problem, which can in turn solves the database join size 
estimation problem for binary relations.

\renewcommand{\alg}{\mathcal{A}}

\newcommand{\ystar}{\ensuremath{y^{*}}}

\newcommand{\estar}{\ensuremath{e^{*}}}

\newcommand{\fstar}{\ensuremath{f^{*}}}

\newcommand{\bsefr}{\bm{0}}

\newcommand{\Disj}{\ensuremath{\textnormal{\textsf{Disj}}}\xspace}

\section{Lower Bounds}\label{sec:lower}

In this section, we prove two separate lower bounds that demonstrate the optimality of Theorem~\ref{thm:main-intro} in different scenarios. Our first lower bound in Section~\ref{sec:lb-odd-cycles} establishes tight bounds for estimating the 
number of \emph{odd cycles}. This result implies that in addition to cliques (that were previously proved~\cite{EdenR18b}; see also 
in~\cite{EdenLRS15,EdenRS18}), our algorithm in Theorem~\ref{thm:main-intro} also achieve optimal bounds for odd cycles. 
Next, in Section~\ref{sec:lb-db-join}, we target the more general problem of database join size estimation 
for which we argued that our Theorem~\ref{thm:main-intro} continues to hold. We show that for this more general problem, our algorithm in Theorem~\ref{thm:main-intro} is in fact optimal \emph{for all} choices of the subgraph 
$H$.

\subsection{A Lower Bound for Counting Odd Cycles}\label{sec:lb-odd-cycles}

We prove that the bound achieved by Theorem~\ref{thm:main-intro} for any odd cycle $C_{2k+1}$ is optimal. 

\begin{theorem}\label{thm:lb-odd}
	For any $k \geq 1$, any algorithm $\alg$ that can output any multiplicative-approximation to the number of copies of the odd cycle $C_{2k+1}$ in a given graph $G(V,E)$ with probability at least $2/3$ 
	requires $\Omega(\frac{m^{k+\frac{1}{2}}}{\sC})$ queries to $G$. 
\end{theorem}
\noindent
Our proof of Theorem~\ref{thm:lb-odd} uses communication complexity in the two player communication model of Yao~\cite{Yao79}. Proving query complexity lower bounds using communication complexity tools in different scenarios 
has a rich history (see, e.g.~\cite{BlaisBM11,BlaisCG17,EdenR18b} and references therein), and was nicely formulated by Eden and Rosenbaum in a recent work~\cite{EdenR18b} for graph estimation problems. 

We prove Theorem~\ref{thm:lb-odd} using a reduction from the \emph{set disjointness} problem in communication complexity. In the set disjointness problem, there are two players Alice and Bob that are given
a bit-string $X \in \set{0,1}^{N}$ and $Y \in \set{0,1}^{N}$, respectively; their goal is to determine whether there exists an index $i \in [N]$ such that $X_i \wedge Y_i = 1$, by communicating a small number of bits between each other (the players have 
access to a {shared} source of random bits, called the \emph{public randomness}). 
A celebrated result in communication complexity, first proved by~\cite{KalyanasundaramS92} and further refined in~\cite{Razborov92,Bar-YossefJKS02}, states that communication complexity of this problem, the minimum
number of bits of communication needed to solve this problem with probability at least $2/3$, is $\Omega(N)$. This lower bound continues to hold even for the special case where we are promised that there exists at 
most one index $i$ such that $X_i \wedge Y_i = 1$. 

\subsection*{The Reduction from Set Disjointness}

For simplicity of exposition, we consider the following variant of set disjointness in which the input to Alice and Bob are two-dimensional arrays $X_{i,j}$ and $Y_{i,j}$ for $(i,j) \in ([{K}] \times [{K}]) \setminus \bigcup_{i'\in [K]} \{(i', i')\}$; the goal now is to 
determine whether there exists $(i,j)$ such that $X_{i,j} \wedge Y_{i,j} = 1$ under the promise that \emph{at most} one such index may exist. We refer to this problem as $\Disj(X,Y)$. It is immediate to verify that communication complexity of $\Disj$ is $\Omega(K^2)$ using a straightforward reduction from the original set disjointness problem (under the aformentioned promise) with $N := K \cdot (K-1)$. 

Fix any algorithm $\alg$ for counting the number of copies of $C_{2k+1}$ to within any multiplicative-approximation factor. We use $\alg$ to design a communication protocol $\Prot_{\alg}$ for solving $\Disj(X,Y)$ for an appropriately 
chosen value of $K$ such that communication cost of the new protocol is within a constant factor of the query complexity of $\alg$. 
To do this, the players construct a graph $G_{X,Y}(V,E)$ (corresponding to inputs $X,Y$ of Alice and Bob) 
\emph{implicitly} and run $\alg$ on $G_{X,Y}$ by answering the queries of $\alg$ on $G_{X,Y}$ through communicating with each together. The graph $G_{X,Y}$ is (implicitly) constructed as follows (see Figure~\ref{fig:lb-odd} for an illustration): 

\begin{enumerate}
	\item Partition the set of vertices $V$ into $(k+1)$ layers $V^1,\ldots,V^{k+1}$  each of size $K$. 
	\item For every $1 < i < k+1$, connect every vertex in layer $V^i$ to every vertex in layer $V^{i+1}$. 
	\item For every $(i,j) \in ([{K}] \times [{K}]) \setminus \bigcup_{i' \in [K]} \set{(i',i')}$, if $X_{i,j} \wedge Y_{i,j} = 1$, there exists an edge $(u^1_i,v^1_j)$ for $u^1_i,v^1_j \in V^1$ and another edge $(u^2_i,v^2_j)$ for $u^2_i,v^2_j \in V^2$.
	\item For every $(i,j) \in (i,j) \in ([{K}] \times [{K}]) \setminus \bigcup_{i\in [K]} \{(i, i)\}$, if $X_{i,j} \wedge Y_{i,j} = 0$, there exists an edge $(u^1_i,v^2_j)$ for $u^1_i\in V^1$ and $v^2_j \in V^2$ and another edge $(u^1_j,v^2_i)$ for $u^1_j \in V^1$ and $v^2_i \in V^2$. 
\end{enumerate}
\noindent
The following figure illustrates the graph $G_{X,Y}$ for the case of $C_{7}$.

\begin{figure}[h!]
    \centering
    \subcaptionbox{When $X_{i,j} \wedge Y_{i,j} = 1$.}[0.45\textwidth]{

    \begin{tikzpicture}[ auto ,node distance =1cm and 2cm , on grid , semithick , state/.style ={ circle ,top color =white , bottom color = white , draw, black , text=black}, every node/.style={inner sep=0,outer sep=0}]

\node[state, ellipse, black, minimum height=100pt, minimum width=25pt] (V1){};
\node[right,fill=white,opacity=0, text opacity=1] (tV1)[above=2.2cm of V1]{$V^1$};
\node[state, ellipse, black, minimum height=100pt, minimum width=25pt] (V2) [right = 2cm of V1]{};
\node[right,fill=white,opacity=0, text opacity=1] (tV2)[above=2.2cm of V2]{$V^2$};
\node[state, ellipse, black,  minimum height=100pt, minimum width=25pt] (V3)[right= 2cm of V2]{};
\node[right,fill=white,opacity=0, text opacity=1] (tV3)[above=2.2cm of V3]{$V^2$};
\node[state, ellipse, black,  minimum height=100pt, minimum width=25pt] (V4)[right= 2cm of V3]{};
\node[right,fill=white,opacity=0, text opacity=1] (tV4)[above=2.2cm of V4]{$V^4$};

\node[state, circle, black, line width=0.25mm, minimum height=7pt, minimum width=7pt] (u1)[above = 0.75cm of V1]{$u^1_i$};
\node[state, circle, black, line width=0.25mm, minimum height=7pt, minimum width=7pt] (v1)[above = -0.75cm of V1]{$v^1_j$};

\node[state, circle, black, line width=0.25mm, minimum height=7pt, minimum width=7pt] (u2)[above = 0.75cm of V2]{$u^2_i$};
\node[state, circle, black, line width=0.25mm, minimum height=7pt, minimum width=7pt] (v2)[above = -0.75cm of V2]{$v^2_j$};

\node[inner sep=5pt, draw, dashed, ForestGreen, fit=(u1) (v2) (u2) (v1)] {};

\tikzset{decoration={snake,amplitude=.4mm,segment length=2mm, post length=0mm,pre length=0mm}}

\draw[line width=0.5mm, decorate] (V2) to (V3);
\draw[line width=0.5mm, decorate] (V4) to (V3);

\draw[line width=0.5mm, blue] (u1) to (v1);
\draw[line width=0.5mm, blue] (u2) to (v2);


\end{tikzpicture}

}
\hspace{0.5cm}   \subcaptionbox{When $X_{i,j} \wedge Y_{i,j} = 0$.}[0.45\textwidth]{

    \begin{tikzpicture}[ auto ,node distance =1cm and 2cm , on grid , semithick , state/.style ={ circle ,top color =white , bottom color = white , draw, black , text=black}, every node/.style={inner sep=0,outer sep=0}]

\node[state, ellipse, black, minimum height=100pt, minimum width=25pt] (V1){};
\node[right,fill=white,opacity=0, text opacity=1] (tV1)[above=2.2cm of V1]{$V^1$};
\node[state, ellipse, black, minimum height=100pt, minimum width=25pt] (V2) [right = 2cm of V1]{};
\node[right,fill=white,opacity=0, text opacity=1] (tV2)[above=2.2cm of V2]{$V^2$};
\node[state, ellipse, black,  minimum height=100pt, minimum width=25pt] (V3)[right= 2cm of V2]{};
\node[right,fill=white,opacity=0, text opacity=1] (tV3)[above=2.2cm of V3]{$V^2$};
\node[state, ellipse, black,  minimum height=100pt, minimum width=25pt] (V4)[right= 2cm of V3]{};
\node[right,fill=white,opacity=0, text opacity=1] (tV4)[above=2.2cm of V4]{$V^4$};

\node[state, circle, black, line width=0.25mm, minimum height=7pt, minimum width=7pt] (u1)[above = 0.75cm of V1]{$u^1_i$};
\node[state, circle, black, line width=0.25mm, minimum height=7pt, minimum width=7pt] (v1)[above = -0.75cm of V1]{$v^1_j$};

\node[state, circle, black, line width=0.25mm, minimum height=7pt, minimum width=7pt] (u2)[above = 0.75cm of V2]{$u^2_i$};
\node[state, circle, black, line width=0.25mm, minimum height=7pt, minimum width=7pt] (v2)[above = -0.75cm of V2]{$v^2_j$};

\node[inner sep=5pt, draw, dashed, ForestGreen, fit=(u1) (v2) (u2) (v1)] {};

\tikzset{decoration={snake,amplitude=.4mm,segment length=2mm, post length=0mm,pre length=0mm}}

\draw[line width=0.5mm, decorate] (V2) to (V3);
\draw[line width=0.5mm, decorate] (V4) to (V3);

\draw[line width=0.5mm, blue] (u1) to (v2);
\draw[line width=0.5mm, blue] (v1) to (u2);


\end{tikzpicture}

}

   \caption{Illustration of the graph $G_{X,Y}$ for the odd cycle $C_7$ and the role of $X_{i,j}$ and $Y_{i,j}$ for some index $(i,j)$ in the choice of edges in $G_{X,Y}$.}
 \label{fig:lb-odd}

\end{figure}
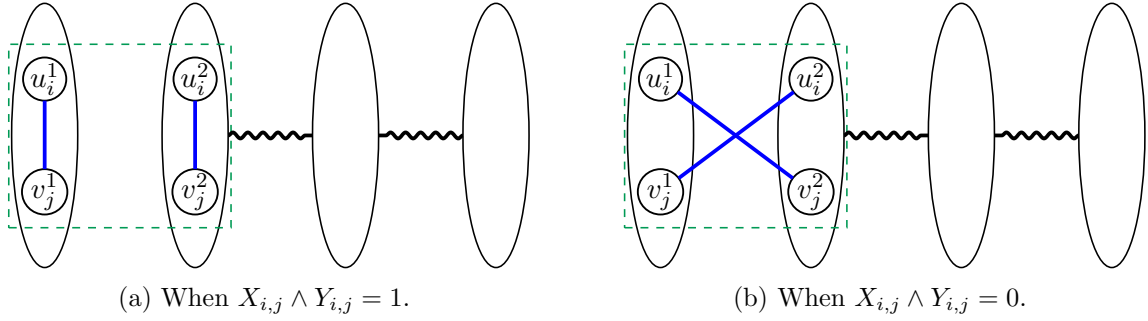


\begin{proposition}\label{prop:reduction}
	For any $X,Y$ with the promise that at most one index $(i,j)$ have $X_{i,j} \wedge Y_{i,j}=1$, in the graph $G_{X,Y}(V,E)$ constructed above: 
	\begin{enumerate}[label=(\roman*)]
		\item The degrees of all vertices are fixed independent of the choice of $X,Y$. 
		\item If for all indices $(i,j)$, $X_{i,j} \wedge Y_{i,j} = 0$, then $\sC = 0$. 
		\item If there exists a unique index $(i,j)$ such that $X_{i,j} \wedge Y_{i,j} = 1$, then $\sC = \Omega\paren{K^{2k-1}}$. 
	\end{enumerate}
\end{proposition}
\begin{proof}
	We prove each part separately: 
	\begin{enumerate}[label=$(\roman*)$]
		\item Follows immediately from the construction (see also Figure~\ref{fig:lb-odd}). 
		\item In this case, all edges of the graph are between $V^i$ and $V^{i+1}$ for some $1 \leq i < k+1$. As such, $G_{X,Y}$ is a bipartite graph with vertices in even layers in one side of the bipartition and the vertices in odd layers in the other side. 
		This means that in this case $G_{X,Y}$ has no odd cycle.  
		\item In this case, there exists a single edge $(u^1_i,v^1_j)$ inside $V^1$. By picking any pair of distinct vertices from $V^2 \setminus \set{u^2_i,v^2_j}$, 
		any pair of distinct vertices from $V^3, \ldots, V^{k}$, a single vertex from $V^{k+1}$, and the vertices $u^1_i,v^1_i$ incident on this edge, we obtain a unique copy of $C_{2k+1}$ in $G_{X,Y}$ (here, we used the assumption that the only edges missing
		between $V^1$ and $V^2$ are $(u^1_i,v^2_j)$ and $(u^1_j,v^2_i)$ by the assumption that at most one index $(i,j)$ has $X_{i,j} \wedge Y_{i,j}=1$).
		As such, 
		\begin{align*}
			\sC = 1 \cdot {{K-2}\choose{2}} \cdot {{K}\choose{2}}^{k-2} \cdot K \geq \paren{\frac{K}{4}}^2 \cdot \paren{\frac{K}{2}}^{2k-4} \cdot K = \Omega(K^{2k-1}),
		\end{align*}
		as $k$ is a constant. 
	\end{enumerate}
	This concludes the proof of Proposition~\ref{prop:reduction}. 
\end{proof}

Now let $\alg$ be a query algorithm for finding any multiplicative-approximation to $C_{2k+1}$ on graphs $G_{X,Y}$ constructed above. By the first part of Proposition~\ref{prop:reduction}, we can safely assume that $\alg$ knows degrees of all vertices in 
$G_{X,Y}$ as degrees of all vertices are always the same. Moreover, any edge-sample query performed by $\alg$ can be instead performed by first sampling one of the vertices proportional to its degree (as all degrees 
are known to $\alg$) and then making a random neighbor query on this vertex. As such, we assume without loss of generality that $\alg$ only performs neighbor and pair queries. 
We now show how to design the protocol $\Prot_{\alg}$ by simulating $\alg$ on the graph $G_{X,Y}$. 
\begin{tbox}
\textbf{The protocol $\Prot_{\alg}$.}

\begin{enumerate}[leftmargin=20pt]
	\item Alice and Bob use public randomness as the random bits needed by $\alg$. 
	\item For every query performed by $\alg$, the players determine the answer to the query on $G_{X,Y}$ as follows, update the state of $\alg$ consistently, and continue to the next query.
	\begin{itemize}[leftmargin=10pt]
		\item \emph{Pair query $(u,v)$:} If $u = u^1_i \in V^1$ and $v=v^2_j \in V^2$ (or vice versa), Alice communicates $X_{i,j}$ to Bob and Bob sends $Y_{i,j}$ to Alice. After this both players can determine the answer to this query by checking whether
		$X_{i,j} \wedge Y_{i,j} = 1$ or not. They do the same when both $u,v$ are in $V^1$ or are in $V^2$. In any other case, the answer to the query is		
		independent of the input to players and they can answer the query with no communication. 
		\item \emph{Neighbor query $(u,j)$:} Suppose $u = u^1_i \in V^1$. If $i=j$, then the answer to the query is $v^2_j \in V^2$. Otherwise, Alice and Bob communicate $X_{i,j}$ and $Y_{i,j}$ and both players determine $X_{i,j} \wedge Y_{i,j}$. If  
		$X_{i,j} \wedge Y_{i,j}=0$, the answer to the query is $v^2_j \in V^2$
		and otherwise it is $v^1_j$ in $V^1$. This is done similarly for when $u = u^2_i \in V^2$. In any other case, the answer to the query 		
		is independent of the input to players and they can answer the query with no communication. 
	\end{itemize}
	\item At the end, if the answer returned by $\alg$ is non-zero, they return that there exists some index $(i,j)$ such that $X_{i,j} \wedge Y_{i,j} = 1$ and otherwise they output no such index exists. 
\end{enumerate}
\end{tbox}

\subsection*{Proof of Theorem~\ref{thm:lb-odd}}

We now prove the correctness of the protocol $\Prot_{\alg}$ in the previous part and establish Theorem~\ref{thm:lb-odd}. 

\begin{proof}[Proof of Theorem~\ref{thm:lb-odd}]
	Let $\alg$ be any query algorithm for counting $C_{2k+1}$ with probability of success at least $2/3$, and let $\Prot_{\alg}$ be the protocol created based on $\alg$. 
	By Proposition~\ref{prop:reduction}, for any input $X,Y$ to $\Disj(X,Y)$ that satisfies the required promise, the graph $G(X,Y)$ contains a copy of $C_{2k+1}$ iff there exists an index $(i,j)$ such that $X_{i,j} \wedge Y_{i,j} = 1$. 
	As such, the output of $\alg$ on $G_{X,Y}$ (whenever correct) is non-zero iff there exists an index $(i,j)$ such that $X_{i,j} \wedge Y_{i,j} = 1$. As the answer returned to each query of $\alg$ in the protocol $\Prot_{\alg}$ is consistent with the 
	underlying graph $G_{X,Y}$, Alice and Bob can simulate $\alg$ on $G_{X,Y}$ correctly and hence their output would be correct with probability at least $2/3$. Additionally, simulating each query access of $\alg$ requires $O(1)$ communication by players
	hence communication cost of $\Prot_{\alg}$ is within constant factor of query complexity of $\alg$. 
	
	Note that the number of edges in the graph $G_{X,Y}$ is $m = \Theta(K^2)$. By the lower bound of $\Omega(K^2)$ on the communication complexity of $\Disj$, we obtain that query cost of $\alg$ needs to be $\Omega(K^2) = \Omega(m)$.
	On the other hand, the last part of Proposition~\ref{prop:reduction} implies that the number of copies of $C_{2k+1}$ in $G$ is $\Omega(m^{k-\frac{1}{2}})$. By re-parametrizing the lower bound of $\Omega(m)$ on the query 
	complexity of $\alg$, we obtain that $\alg$ needs to make at least $\Omega(\frac{m^{k+\frac{1}{2}}}{\sC})$, finalizing the proof. 
\end{proof}

\subsection{A Lower Bound for Database Join Size Estimation}\label{sec:lb-db-join}

Recall the colorful subgraph counting problem (the abstraction of database join size estimation problem) from Section~\ref{sec:extensions}. We prove the following theorem in this section. 

\begin{theorem}\label{thm:lb-colorful}
	For any subgraph $H(V_H,E_H)$ which contains at least one edge,
	suppose $\alg$ is an algorithm for the colorful subgraph estimation problem that given $H$, a coloring $c_H: E_H \rightarrow \IN$, and query access to $G(V,E)$ with $m$ edges and coloring function 
	$c_G:E \rightarrow \IN$, can output a multiplicative-approximation to the number of colorful copies of $H$ in $G$ with probability at least $2/3$.  Then, $\alg$
	requires $\Omega(\frac{m^{\rho(H)}}{\sH_{c}})$ queries, where $\sH_c$ is the number of colorful copies of $H$ in $G$. The lower bound 
	continues to hold even if the number of colors used by $c_H$ and $c_G$ is at most two. 
\end{theorem}

Recall the fractional edge-cover LP in of Section~\ref{sec:edge-cover} (see LP~(\ref{lp:ec})). The following linear program for fractional independent-set is the dual to the edge-cover LP (and hence by LP duality has the same optimal value): 
\begin{align}
\rho(H)\quad = \quad&\textnormal{maximize} && \hspace{-2.5cm} \textnormal{$\sum_{a \in V(H)}y_a$} \notag \\
&\textnormal{subject to} && \hspace{-2.25cm} \textnormal{$y_a + y_b \leq 1$ for all edges $(a,b) \in E(H)$.} \label{lp:is}
\end{align}
\noindent
Throughout this section, we fix an optimal solution $\ystar$ of LP~(\ref{lp:is}). We use $\ystar$ to design two distributions $\FG_0$ and $\FG_1$ on graphs $G$ with $O(m)$ edges\footnote{For simplicity of exposition, we let the graphs contain $O(m)$ edges instead of exactly $m$ edges (but provide the algorithm with the exact number of edges in the graph); a simple rescaling of the bound immediately proves the lower bound for the case of graphs with exactly $m$ edges as well.} such that any graph $G$ sampled from $\FG_0$, denoted by $G \sim \FG_0$, contains no colorful copy of $H$ (for a specific 
coloring of $H$ to be described later), while any $G \sim \FG_1$ contains many colorful copies of $H$. We then prove that any algorithm that makes only a small number of queries to the underlying graph cannot distinguish between graphs sampled 
from $\FG_0$ and $\FG_1$, concluding the proof. 

\subsection*{Distributions $\FG_0$ and $\FG_1$}

We first define the coloring $c_H$ of $H$. Let $\fstar:= (a,b)$ be any arbitrary edge in $H$ such that $\ystar_a + \ystar_b = 1$, i.e., is a tight constraint for $\ystar$ in LP~(\ref{lp:is}). By optimality of $\ystar$ and as $H$ is not a singleton vertex, such
an edge $\fstar$ always exists. We now define $c_H(\fstar) := 1$ and $c_H(f) := 0$ for any other edge $f \in E(H) \setminus \set{\fstar}$. 

We now define the distribution $\FG_0$. In fact, distribution $\FG_0$ has all its mass on a single graph $G_0$ with coloring $c_{G_0}$ which contains no colorful copy of $H$ (under the coloring $c_H$ defined above).  
Suppose $H$ has $k \geq 2$ vertices denoted by $V(H) := \set{a_1,\ldots,a_k}$. The graph $G_0$ is constructed as follows. Firstly, the vertices of $G_0$ are partitioned into $k$ sets $V(G_0) := V_1 \cup \ldots \cup V_k$ with $\card{V_i} = m^{\ystar_{a_i}}$. 
Then for any edge $(a_i,a_j) \in E(H)$, we connect all vertices in $V_i$ to all vertices in $V_j$ in $G_0$. Finally, the coloring $c_{G_0}$ of $G_0$ simply assigns the color $0$ to \emph{all} edges in $G_0$. See Figure~\ref{fig:lb-db} for an illustration.

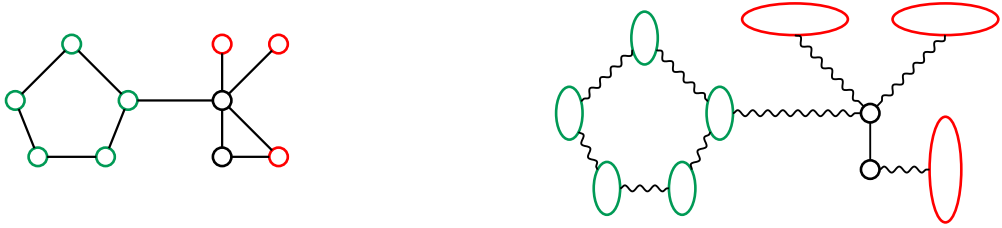
\begin{figure}[h!]
    \centering
    \subcaptionbox{The subgraph $H$. The number next to each vertex $a$ denotes $\ystar_a$.}[0.35\textwidth]{

    \begin{tikzpicture}[ auto ,node distance =1cm and 2cm , on grid , semithick , state/.style ={ circle ,top color =white , bottom color = white , draw, black , text=black}, every node/.style={inner sep=0,outer sep=0}]

\node[state, circle, ForestGreen, line width=0.35mm, minimum height=7pt, minimum width=7pt] (a1){};
\node[right,fill=white,opacity=0, text opacity=1] (ta1)[above=0.3cm of a1]{\textcolor{ForestGreen}{$0.5$}};

\node[state, circle, ForestGreen, line width=0.35mm, minimum height=7pt, minimum width=7pt] (a2) [below left=0.75cm and 0.75cm of a1]{};
\node[right,fill=white,opacity=0, text opacity=1] (ta2)[above left=0.3cm and 0.1cm of a2]{\textcolor{ForestGreen}{$0.5$}};
\node[state, circle, ForestGreen, line width=0.35mm, minimum height=7pt, minimum width=7pt] (a3) [below right=0.75cm and 0.75cm of a1]{};
\node[right,fill=white,opacity=0, text opacity=1] (ta3)[above right=0.3cm and 0.1cm of a3]{\textcolor{ForestGreen}{$0.5$}};
\node[state, circle, ForestGreen, line width=0.35mm, minimum height=7pt, minimum width=7pt] (a4) [below right=0.75cm and 0.3cm of a2]{};
\node[right,fill=white,opacity=0, text opacity=1] (ta4)[ left=0.4cm of a4]{\textcolor{ForestGreen}{$0.5$}};
\node[state, circle, ForestGreen, line width=0.35mm, minimum height=7pt, minimum width=7pt] (a5) [below left=0.75cm and 0.3cm of a3]{};
\node[right,fill=white,opacity=0, text opacity=1] (ta5)[ right=0.4cm of a5]{\textcolor{ForestGreen}{$0.5$}};

\node[state, circle, white, line width=0.35mm, minimum height=7pt, minimum width=7pt] (x) [below left=0.75cm and 0.3cm of a5]{};

\node[state, circle, black, line width=0.35mm, minimum height=7pt, minimum width=7pt](b1) [below right=0.75cm and 2cm of a1]{};
\node[right,fill=white,opacity=0, text opacity=1] (tb1)[above left=0.25cm and 0.2cm of b1]{$0$};
\node[state, circle, black, line width=0.35mm, minimum height=7pt, minimum width=7pt] (b2) [below=0.75cm of b1]{};
\node[right,fill=white,opacity=0, text opacity=1] (tb2)[left=0.3cm of b2]{$0$};
\node[state, circle, red, line width=0.35mm, minimum height=7pt, minimum width=7pt](b3) [right=0.75cm of b2]{};
\node[right,fill=white,opacity=0, text opacity=1] (tb3)[right=0.3cm of b3]{\textcolor{red}{$1$}};
\node[state, circle, red, line width=0.35mm, minimum height=7pt, minimum width=7pt] (b4) [above=0.75cm of b1]{};
\node[right,fill=white,opacity=0, text opacity=1] (tb4)[above=0.3cm of b4]{\textcolor{red}{$1$}};
\node[state, circle, red, line width=0.35mm, minimum height=7pt, minimum width=7pt] (b5) [right=0.75cm of b4]{};
\node[right,fill=white,opacity=0, text opacity=1] (tb5)[above=0.3cm of b5]{\textcolor{red}{$1$}};

\draw[line width=0.3mm](a1) to (a2);
\draw[line width=0.3mm] (a1) to (a3);
\draw[line width=0.3mm] (a2) to (a4);
\draw[line width=0.3mm] (a4) to (a5);
\draw[line width=0.3mm] (a5) to (a3);

\draw[line width=0.3mm] (b1) to (b2);
\draw[line width=0.3mm] (b1) to (b3);
\draw[line width=0.3mm] (b2) to (b3);
\draw[line width=0.3mm] (b1) to (b4);
\draw[line width=0.3mm] (b1) to (b5);

\draw[line width=0.3mm] (a3) to (b1);

\end{tikzpicture}

}
\hspace{0.5cm}   \subcaptionbox{The graph $G_0$ of $\FG_0$. The number next to each block of vertices denotes the size of the block. }[0.55\textwidth]{

    \begin{tikzpicture}[ auto ,node distance =1cm and 2cm , on grid , semithick , state/.style ={ circle ,top color =white , bottom color = white , draw, black , text=black}, every node/.style={inner sep=0,outer sep=0}]

\node[state, ellipse, ForestGreen, line width=0.35mm, minimum height=20pt, minimum width=10pt] (a1){};
\node[right,fill=white,opacity=0, text opacity=1] (ta1)[above=0.6cm of a1]{\textcolor{ForestGreen}{$\sqrt{m}$}};
\node[state, ellipse, ForestGreen, line width=0.35mm, minimum height=20pt, minimum width=10pt] (a2) [below left=1cm and 1cm of a1]{};
\node[right,fill=white,opacity=0, text opacity=1] (ta2)[above left=0.6cm and 0.2cm of a2]{\textcolor{ForestGreen}{$\sqrt{m}$}};
\node[state, ellipse, ForestGreen, line width=0.35mm, minimum height=20pt, minimum width=10pt] (a3) [below right=1cm and 1cm of a1]{};
\node[right,fill=white,opacity=0, text opacity=1] (ta3)[above right=0.6cm and 0.2cm of a3]{\textcolor{ForestGreen}{$\sqrt{m}$}};
\node[state, ellipse, ForestGreen, line width=0.35mm, minimum height=20pt, minimum width=10pt] (a4) [below right=1cm and 0.5cm of a2]{};
\node[right,fill=white,opacity=0, text opacity=1] (ta4)[ left=0.6cm of a4]{\textcolor{ForestGreen}{$\sqrt{m}$}};
\node[state, ellipse, ForestGreen, line width=0.35mm, minimum height=20pt, minimum width=10pt] (a5) [below left=1cm and 0.5cm of a3]{};
\node[right,fill=white,opacity=0, text opacity=1] (ta5)[ right=0.6cm of a5]{\textcolor{ForestGreen}{$\sqrt{m}$}};

\node[state, circle, black, line width=0.35mm, minimum height=7pt, minimum width=7pt](b1) [right=2cm of a3]{};
\node[right,fill=white,opacity=0, text opacity=1] (tb1)[below right=0.25cm and 0.2cm of b1]{$1$};
\node[state, ellipse, black, line width=0.35mm, minimum height=7pt, minimum width=7pt] (b2) [below=0.75cm of b1]{};
\node[right,fill=white,opacity=0, text opacity=1] (tb2)[left=0.3cm of b2]{$1$};
\node[state, ellipse, red, line width=0.35mm, minimum height=40pt, minimum width=12pt](b3) [right=1cm of b2]{};
\node[right,fill=white,opacity=0, text opacity=1] (tb3)[right=0.45cm of b3]{\textcolor{red}{${m}$}};
\node[state, ellipse, red, line width=0.35mm, minimum height=12pt, minimum width=40pt] (b4) [above left =1.25cm and 1cm of b1]{};
\node[right,fill=white,opacity=0, text opacity=1] (tb4)[above=0.45cm of b4]{\textcolor{red}{${m}$}};
\node[state, ellipse, red, line width=0.35mm, minimum height=12pt, minimum width=40pt] (b5) [right=2cm of b4]{};
\node[right,fill=white,opacity=0, text opacity=1] (tb5)[above=0.45cm of b5]{\textcolor{red}{${m}$}};

  \tikzset{decoration={snake,amplitude=.4mm,segment length=2mm,
                       post length=0mm,pre length=0mm}}
                       
\draw[line width=0.25mm, decorate] (a1) to (a2);
\draw[line width=0.25mm, decorate] (a1) to (a3);
\draw[line width=0.25mm, decorate] (a2) to (a4);
\draw[line width=0.25mm, decorate] (a4) to (a5);
\draw[line width=0.25mm, decorate] (a5) to (a3);

\draw[line width=0.25mm,decorate] (a3) to (b1);

\draw[line width=0.25mm,decorate] (b4.south) to (b1);

\draw[line width=0.25mm,decorate] (b5.south) to (b1);

\draw[line width=0.25mm] (b1) -- (b2);

\draw[line width=0.25mm, decorate] (b2) -- (b3);


\end{tikzpicture}

}

   \caption{Illustration of the graph $G_0$ in distribution $\FG_0$.}
 \label{fig:lb-db}

\end{figure}

The distribution $\FG_1$ is constructed similarly (but on a larger support). Let $G_0$ be the single graph constructed by $\FG_0$. Any graph $G \sim \FG_1$ is constructed as follows: we first let $G = G_0$ and then choose a single edge $\estar$ uniformly
at random from the edges between $V_i$ and $V_j$ where $(i,j)$ is chosen such that $\fstar = (a_i,a_j)$ (recall the definition of edge $\fstar$ above). We then change the color $c_G(\estar) = 1$ (all other edges are still assigned the color $0$). 
This concludes the description of distributions $\FG_0$ and $\FG_1$. 
We now present basic properties of these distributions. 

\begin{proposition}\label{prop:dist}
	For the two distributions $\FG_0$ and $\FG_1$:
	\begin{enumerate}[label=(\roman*)]
		\item\label{p:dist1} Every graph $G$ sampled from $\FG_0$ or $\FG_1$ contains $\Theta(m)$ edges. 
		\item\label{p:dist2} The graph $G_0 \sim \FG_0$ contains no colorful copies of $H$, while any graph $G \sim \FG_1$ contains $m^{\rho(H)-1}$ colorful copies of $H$. 
		\item\label{p:dist3} For a graph $G \sim \FG_1$, the edge $\estar$ is chosen uniformly at random among the $m$ edges between $V_i$ and $V_j$. 
	\end{enumerate}
\end{proposition}
\begin{proof} We prove each part separately below. 
	\begin{enumerate}[label=($\roman*$)]
	\item The number of edges sampled from the distributions $\FG_0$ and $\FG_1$ is the same, hence it only suffices to prove the bound for the (unique) graph $G_0$ in the support of $\FG_0$. For any edge $(a_i,a_j)$ in $H$, 
	we have a bipartite clique between $V_i$ and $V_j$ in $G_0$, hence resulting in $\card{V_i} \cdot \card{V_j} = m^{\ystar_{a_i}} \cdot m^{\ystar_{a_j}} \leq m$ edges in $G$, where the final inequality is because $\ystar$ is a feasible
	solution of LP~(\ref{lp:is}). As such, the number of edges in $G_0$ is $O(m)$ as size of $H$ is constant. 
	\item There is no edge with color $1$ in $G_0 \sim \FG_0$, while $H$ has an edge with color $1$ and hence $G_0$ contains no colorful copy of $H$. On the other hand, in any graph $G \sim \FG_1$, we can create a copy of $H$
	by mapping each vertex $a_i$ of $V(H)$ which is not incident to $\fstar$ to any arbitrary vertex in $V_i$ and then maping the edge $\fstar$ of $H$ to $\estar$ in $G$. Suppose $\fstar=(a,b)$. The total number of colorful copies of $H$ in $G$ is 
	then $\prod_{a_i \in V(H) \setminus \set{a,b}} \card{V_i} = m^{\sum_{a_i \in V(H) \setminus \set{a,b}} \ystar_{a_i}} = m^{\rho(H)-1}$ as $\sum_{a_i \in V(H)} \ystar_{a_i} = \rho(H)$ and $\ystar_{a} + \ystar_{b} = 1$. 
	\item The fact that $\estar$ is chosen uniformly at random is by definition of distribution $\FG_1$. The total number of edges between $V_i$ and $V_j$ where $\estar$ is chosen from is $\card{V_i} \cdot \card{V_j} = m^{\ystar_{a_i} + \ystar_{a_j}} = m$
	by the choice of $\fstar = (a_i,a_j)$. 
	\end{enumerate}
	This concludes the proof of Proposition~\ref{prop:dist}.
\end{proof}

\subsection*{Query Complexity of Distinguishing $\FG_0$ and $\FG_1$}

We now prove that any query algorithm that can distinguish between instances sampled from $\FG_0$ and $\FG_1$ requires $\Omega(m)$ queries, proving the following lemma. 

\begin{lemma}\label{lem:qc-fg}
	Define the distribution $\FG := \frac{1}{2} \cdot \FG_0 + \frac{1}{2} \cdot \FG_1$. Suppose $\alg$ is any algorithm that given a graph $G \sim \FG$ with probability at least $2/3$ determines whether it belongs to (the support of) $\FG_0$ or $\FG_1$. 
	 Then $\alg$ needs to make $\Omega(m)$ queries to the graph. 
\end{lemma}
\begin{proof}	
	We assume that $\alg$ knows the partitioning of vertices of $G$ into $V_1,\ldots,V_{\card{V(H)}}$ and is hence even aware of the set of edges in $G$ (but not their colors); this can only strengthen our lower bound. 
	
	Assume $\fstar = (a_i,a_j)$ and note that the only difference between the graphs in $\FG_0$ and $\FG_1$ is that the latter graphs have an edge $\estar$ between $V_i$ and $V_j$ that is colored $1$ instead of $0$. 
	This implies that the only ``useful'' queries performed by $\alg$ are pair queries between vertices $u \in V_i$ and $v \in V_j$ (degree queries can be
	answered without querying the graph; neighbor queries can be simulated by a pair query as the set of neighbors are all known in advance; edge-sample queries can also be performed by pair queries by sampling one of the known edges uniformly at random 
	and then querying the edge to determine its color). 
	
	Suppose towards a contradiction that $\alg$ is an algorithm (possibly randomized) that given a graph $G \sim \FG$ uses $o(m)$ queries and 
	can determine whether $G$ belongs to $\FG_0$ or $\FG_1$ with probability at least $2/3$. By fixing the randomness of this algorithm
	and an averaging argument (namely, the easy direction of Yao's minimax principle~\cite{Yao83}), we obtain a deterministic algorithm $\alg'$ that uses the same number of queries as $\alg$ and output the correct answer with probability $2/3$, where
	the probability is now only taken over the randomness of the distribution $\FG$. 
	
	Let $Q := (q_1,q_2,\ldots,q_\ell)$ for $\ell = o(m)$ determines the (potentially adaptively chosen) set of queries performed by $\alg'$ before it outputs the answer. Since the set of edges in the graph are already known to $\alg$, 
	the only interesting part of the answer to each query $q_i$ is whether the color
	of the edge queried by $q_i$ is $0$ or $1$. With a slight abuse of notation, we write $q_i = 1$ if the color of the edge queried by $q_i$ is $1$ and $q_i = 0$ otherwise. 
	
	Notice that since $\alg'$ is a deterministic algorithm, the next query $q_i$ is determined solely based on the answer to queries $q_1,\ldots,q_{i-1}$. Let $\bsefr_k$ denote the vector of all zeros of length $k$. As a result, 
	\begin{align*}
		\Pr_{G \sim \FG_1} \Bracket{q_i =1 \mid (q_1,\ldots,q_{i-1}) = \bsefr_{i-1}} = \frac{1}{m-i+1}. 
	\end{align*}
	This is because, conditioned on all $(q_1,\ldots,q_{i-1}) = \bsefr_{i-1}$, the next query chosen by $\alg'$ is fixed beforehand and is only based on the knowledge that the $i-1$ edges queried so far cannot be $\estar$. As $\estar$ is chosen uniformly at random
	from a set of $m$ edges (by Part~\ref{p:dist3} of Proposition~\ref{prop:dist}), the bound above holds (note that we assumed without loss of generality that $\alg'$ does not query an  edge more than once). As a result of this, 
	we have, 
	\begin{align}
		\Pr_{G \sim \FG_1} \Bracket{(q_1,\ldots,q_{\ell}) = \bsefr_\ell} = \frac{m-1}{m} \cdot \frac{m-2}{m-1} \cdot \ldots \cdot \frac{m-\ell}{m-\ell-1} = 1- \frac{\ell}{m}. \label{eq:ell/m}
	\end{align}
	\noindent
	Let $O(q_1,\ldots,q_\ell) \in \set{0,1}$ denote the output of $\alg'$ based on the answers given to the queries $q_1,\ldots,q_{\ell}$. We have, 
	\begin{align}
		\Pr_{G \sim \FG} \Bracket{\text{$\alg'$ is correct on $G$}} &= \frac{1}{2} \cdot \Pr_{G \sim \FG_0} \Bracket{O(q_1,\ldots,q_{\ell}) = 0} + \frac{1}{2} \cdot \Pr_{G \sim \FG_1} \Bracket{O(q_1,\ldots,q_{\ell}) = 1} \label{eq:alg'-lb}
	\end{align}
	The second term in RHS above can be upper bounded by,
	\begin{align*}
		\Pr_{G \sim \FG_1} \Bracket{O(q_1,\ldots,q_{\ell}) = 1} &\leq \Pr_{G \sim \FG_1} \Bracket{(q_1,\ldots,q_{\ell}) = \bsefr_\ell} \cdot \Pr_{G \sim \FG_1} \Bracket{O(q_1,\ldots,q_{\ell}) = 1 \mid (q_1,\ldots,q_{\ell}) = \bsefr_\ell} \\
		&\hspace{2.5cm} + \Paren{1-\Pr_{G \sim \FG_1} \Bracket{(q_1,\ldots,q_{\ell}) = \bsefr_\ell}} \\
		&= \paren{1-\frac{\ell}{m}} \cdot \Pr_{G \sim \FG_1} \Bracket{O(q_1,\ldots,q_{\ell}) = 1 \mid (q_1,\ldots,q_{\ell}) = \bsefr_\ell} + \frac{\ell}{m},
	\end{align*}	
	by Eq~(\ref{eq:ell/m}). Plugging in this bound in Eq~(\ref{eq:alg'-lb}) implies that, 
	\begin{align*}
		\Pr_{G \sim \FG} \Bracket{\text{$\alg'$ is correct on $G$}} &\leq \frac{1}{2} \cdot \Pr_{G \sim \FG_0} \Bracket{O(q_1,\ldots,q_{\ell}) = 0} \\
		&\hspace{2.5cm} + \frac{1}{2} \cdot \Pr_{G \sim \FG_1} \Bracket{O(q_1,\ldots,q_{\ell}) = 1 \mid (q_1,\ldots,q_{\ell}) = \bsefr_\ell} + \frac{\ell}{2m}.
	\end{align*}
	We argue that either $\Pr_{G \sim \FG_1} \Bracket{O(q_1,\ldots,q_{\ell}) = 1 \mid (q_1,\ldots,q_{\ell}) = \bsefr_\ell}$ or $\Pr_{G \sim \FG_0}\Bracket{O(q_1,\ldots,q_{\ell}) = 0}$ must be $0$. This is because in both cases, $(q_1,\ldots,q_{\ell}) = \bsefr_{\ell}$
	and hence $O(q_1,\ldots,q_{\ell})$ is fixed to be either $0$ or $1$ at this point. As a result, 
	\begin{align*}
		\Pr_{G \sim \FG} \Bracket{\text{$\alg'$ is correct on $G$}} \leq \frac{1}{2} + \frac{\ell}{2m} = \frac{1}{2} + o(1). 
	\end{align*}
	This contradicts the fact that $\alg'$ outputs the correct answer with probability at least $2/3$, implying that $\ell$ needs to be $\Omega(m)$. 
\end{proof}

\subsection*{Proof of Theorem~\ref{thm:lb-colorful}}

We can now finalize the proof of Theorem~\ref{thm:lb-colorful} using Proposition~\ref{prop:dist} and Lemma~\ref{lem:qc-fg}. 

\begin{proof}[Proof of Theorem~\ref{thm:lb-colorful}]
	Firstly, any algorithm that can provide any multiplicative-approximation to the number of colorful copies of $H$ in graphs $G$ must necessarily distinguish between the graphs chosen from distributions $\FG_0$ and $\FG_1$
	because by Part~\ref{p:dist1} of Proposition~\ref{prop:dist}, graphs in $\FG_0$ contain no colorful copies of $H$ while graphs in $\FG_1$ contain $m^{\rho(H)-1}$ colorful copies of $H$. Moreover, in the graphs chosen 
	from $\FG_1$, $\sH_c = m^{\rho(H)-1}$. The lower bound of $\Omega(m^{\rho(H)}/\sH_c)$ on the query complexity of algorithms now follows from the $\Omega(m)$ lower bound of Lemma~\ref{lem:qc-fg}. 
\end{proof}

\subsection*{Acknowledgements}
We are thankful to the anonymous reviewers of ITCS 2019 for many valuable comments.

\bibliographystyle{abbrv}
\bibliography{general}

\appendix

\section{Missing Details and Proofs}\label{app:appendix}

\subsection{Proof of Proposition~\ref{prop:min-degree}}\label{app:min-degree}

Proposition~\ref{prop:min-degree} (restated here for convenience of the reader) follows from standard graph theory 
facts (see, e.g. Lemma~2 in~\cite{ChibaN85}). We give a self-contained proof here for completeness. 

\begin{proposition*}[Proposition~\ref{prop:min-degree} in Section~\ref{sec:prelim}]
	For any graph $G$, $\sum_{(u,v) \in E} \min(d_u,d_v) \leq 5m\sqrt{m}$. 
\end{proposition*}
\begin{proof}
	Let $V^+$ be the set of vertices with degree more than $\sqrt{m}$ and $V^- := V \setminus V^+$. 
	\begin{align*}
		\sum_{(u,v) \in E} \min(d_u,d_v) &= \frac{1}{2} \cdot \sum_{u \in V}\sum_{v \in N(u)} \min(d_u,d_v) \leq \frac{1}{2} \cdot \paren{\sum_{u \in V^-} \sum_{v \in N(u)} \sqrt{m} + \sum_{u \in V^+} \sum_{v \in N(u)} \min(d_u,d_v)} \\
		&\leq m\sqrt{m} + \frac{1}{2} \cdot \sum_{u \in V^+} \paren{\sum_{v \in N(u): d_{v} < d_{u}} d_v + \sum_{v \in N(u): d_v \geq d_u} d_u} \\
		&\leq m\sqrt{m} + \frac{1}{2} \cdot \sum_{u \in V^+} \paren{2m + \frac{2m}{d_u} \cdot d_u} \leq m\sqrt{m} + \sqrt{m} \cdot 4m,
	\end{align*}
	where the second last inequality is because, sum of degrees of vertices in $N(u)$ is $\leq 2m$, and number of vertices with degree more than $d_{u}$ is $\leq 2m/d_{u}$ and the last inequality is by $\card{V^+} \leq 2\sqrt{m}$. 
\end{proof}

\subsection{Proof of Lemma~\ref{lem:subgraph-decomposition}}\label{app:subgraph-decomposition}

We now provide a self-contained proof of Lemma~\ref{lem:subgraph-decomposition} (restated below) for completeness. 

\begin{lemma*}[Lemma~\ref{lem:subgraph-decomposition} in Section~\ref{sec:edge-cover}]
	Any subgraph $H$ admits an optimal fractional edge-cover $x^*$ such that the support of $x^*$, denoted by $\supp{x^*}$, is a collection of vertex-disjoint
	odd cycles and star graphs, and, 
	\begin{enumerate}
		\item for every odd cycle $C \in \supp{x^*}$, $x^*_e = 1/2$ for all $e \in C$;
		\item for every edge $e \in \supp{x^*}$ that does \emph{not} belong to any odd cycle,  $x_e = 1$. 
	\end{enumerate}
\end{lemma*}

To prove Lemma~\ref{lem:subgraph-decomposition}, we first state a basic property of LP~(\ref{lp:ec}). 
\begin{proposition}\label{prop:half-integrality-ec}
	LP~(\ref{lp:ec}) admits a half-integral optimum solution $x^* \in \set{0,\frac{1}{2},1}^{\card{E(H)}}$.  Moreover, if $H$ is bipartite, then LP~(\ref{lp:ec}) admits an integral optimum solution. 
\end{proposition}
\begin{proof}
	Suppose first that $H$ is bipartite and $x \in [0,1]^{\card{E(H)}}$ is some optimal solution of LP~(\ref{lp:ec}). We perform a simple cycle-canceling on $x$ to make it integral. In particular, let 
	$e_1,\ldots,e_{2k}$ for some integer $k \geq 2$ be a cycle in the support of $x$ (as $H$ is bipartite length of this cycle is necessarily even). We can alternatively increase the value on one edge and decrease the value
	on the next one by the same amount and continue along the cycle until the value on an edge drops to zero. This operation clearly preserves the feasibility as well as the value of the solution. By doing this, we can cancel all cycles in the support of $x$ 
	without changing the value of LP or violating the feasibility. At this point, support of $x$ is a forest and can be turned into an integral solution using a standard deterministic rounding in a bottom up approach from the leaf-nodes of the forest (see the proof 
	of Lemma~\ref{lem:subgraph-decomposition} for more details on this standard procedure).  
	
	Now suppose $H$ is a non-bipartite graph. Create the following bipartite graph $H'$ where $V(H')$ consists of two copies of vertices in $H$, i.e., for any vertex $a \in V(H)$, there are two copies, say, $a^L$ and $a^R$ in $V(H')$. 
	Moreover, for any edge $e:=(a,b) \in E(H)$ there are two edges $e_1:=(a^L,b^R)$ and $e_2:=(a^R,b^L)$ in $E(H')$. It is easy to see that any edge cover $y$ of $H'$ can be translated to an edge cover $x$ of $H$ by
	setting $x_e = \frac{y_{e_1}+y_{e_2}}{2}$. As by the first part, $H'$ admits an integral optimum solution, we immediately have that $H$ admits a half-integral optimum solution.
\end{proof}

\begin{proof}[Proof of Lemma~\ref{lem:subgraph-decomposition}]
	Let $x^*$ be a half-integral optimum solution for the graph $H$ that is guaranteed to exist by Proposition~\ref{prop:half-integrality-ec}. Let $C$ be any cycle (odd or even length) in $\supp{x^*}$. For any edge $e \in C$, 
	$x^*_e = 1/2$ as otherwise by decreasing $x^*_e$ from $1$ to $1/2$ (recall that $x^*$ is half-integral and $x^*_e \neq 0$), we can reduce the optimal solution without violating the feasibility. Moreover, if 
	$C$ is of even length, then we can perform a standard cycle canceling (by adding and subtracting $1/2$ to the value of $x^*$ on the alternate edges of $C$)
	and remove the cycle. Now suppose $C$ is of odd length; we argue that for each vertex $a \in C$, the only edges in $\supp{x^*}$ that are incident on $a$ are edges in $C$.	
	
	Suppose by contradiction that there exists an edge $e$ with $x^*_e \geq 1/2$ which is incident on a vertex $a$ in an odd cycle $C$ (in $\supp{x^*}$). Perform a cycle canceling as follows: subtract $1/2$ 
	from every other edge starting from an edge incident to $a$ and add $1/2$ to every other edge \emph{plus} the edge $e$. As $C$ is an odd cycle, the total number of addition and subtractions are equal and hence 
	does not change the value of $x^*$. It is also easy to verify that the new $x^*$ is still feasible as $x^*_e \geq 1$ now and hence $a$ is covered still. 
	Thus, by repeatedly applying the above argument, we can change $x^*$ so that $\supp{x^*}$ consists of a vertex-disjoint union of odd cycles (with $x^*_e = 1/2$) and forests. We now turn the forests into a collection of 
	starts using a simple deterministic rounding. 
	
	For each tree $T$ in this forest, we root the tree arbitrarily at some degree one vertex.  Any edge $e$ incident on leaf-nodes of this tree clearly has $x^*_e = 1$.  Let $f$ be a parent edge $e$ and $z$ be
	a parent of $f$ (if these edges do not exist, $e$ belongs to a cycle and we are already done). Let $x^*_z \leftarrow x^*_z + x^*_f$ and $x^*_f \leftarrow 0$. This preserves both the value of $x^*$ and its feasibility, and further partition this 
	tree into a forest and a star. By repeatedly applying this argument, we can decompose every forest into a collection of stars, finalizing the proof of  Lemma~\ref{lem:subgraph-decomposition}. 
\end{proof}

\subsection{An Alternate Analysis of the Variance of the Estimator for Stars}\label{app:star-variance}

Recall that in Lemma~\ref{lem:ssestimator-output}, we upper bounded the variance of the random variable $X$ associated with $\ssestimator(G,S_\ell)$ with $\var{X} \leq 2m^{\ell} \cdot \expect{X}$.               
Using this analysis in our Theorem~\ref{thm:alg-estimate} results in an upper bound of $O(\frac{m^{\ell}}{\sS})$ on the query complexity of counting stars which is suboptimal. We now show that a 
slightly improved analysis of the variance in fact results in an algorithm with $\Os(\frac{m}{(\sS)^{1/\ell}})$ query complexity which is optimal by a result of~\cite{AliakbarpourBGP18}. 

\begin{lemma}\label{lem:ssestimator-better}
	For the random variable $X$ associated with $\ssestimator(G,S_\ell)$, 
	\begin{align*}
		\expect{X} &= (\sS), \qquad 
		\var{X} \leq 4m \cdot \ell^{2\ell} \cdot (\sS)^{2-1/\ell}.
	\end{align*}
\end{lemma}
\begin{proof}
	The bound on the expectation is already establish in Lemma~\ref{lem:ssestimator-output}. We now prove the bound on variance. 
	This proves the desired bound on the exception. We now bound $\var{X}$.
	\begin{align*}
		\var{X} &\leq \expect{X^2} =  \sum_{v \in V}\sum_{\substack{\bw  \in N(v)^{\ell}}}\Pr\paren{\Label{\alpha_r}=v} \cdot \Pr\paren{\Label{\alpha_l} = \bw} \\
		 &\hspace{4.5cm} \cdot \II(\text{$(v,\bw)$ forms a copy of $S_{\ell}$}) \cdot \paren{\Value{\alpha_r} \cdot \Value{\alpha_l}}^2 \\
		&= \sum_{v }\frac{d_v}{2m} \cdot \sum_{\bw} \frac{1}{{{d_v}\choose{\ell}}}  \cdot \II(\text{$(v,\bw)$ forms a copy of $S_{\ell}$}) \cdot \paren{(2m/d_v) \cdot {{d_v}\choose{\ell}}}^2 \\
		&= \sum_{v } \sum_{\bw} \II(\text{$(v,\bw)$ forms a copy of $S_{\ell}$}) \cdot (2m/d_v) \cdot {{d_v}\choose{\ell}} \\
		&\leq 2m \cdot  \sum_{v} \sum_{\bw} \II(\text{$(v,\bw)$ forms a copy of $S_{\ell}$}) \cdot {d_v}^{\ell-1} \tag{since ${{d_v}\choose{\ell}} \leq {d_v}^{\ell}$}\\
		&= 2m \cdot  \sum_{v: d_v \geq \ell} {{d_v}\choose{\ell}} \cdot {d_v}^{\ell-1} \tag{since $\sum_{\bw} \II(\text{$(v,\bw)$ forms a copy of $S_{\ell}$}) = {{d_v}\choose{\ell}}$ for any $v$ with $d_v \geq \ell$ and is $0$ otherwise}\\
		&= 2m \cdot  \sum_{v: d_v \geq \ell} \paren{{d_v}^{\ell}}^{2-1/\ell} \tag{since ${{d_v}\choose{\ell}} \leq {d_v}^{\ell}$}\\
		&\leq 2m \cdot  \Paren{\sum_{v: d_v \geq \ell} {d_v}^{\ell}}^{2-1/\ell} \tag{since $\sum_{a} a^b \leq \paren{\sum_a a}^{b}$ for $b \geq 1$}\\
		&\leq 2m \cdot  \Paren{\sum_{v: d_v \geq \ell} \ell^\ell \cdot {{d_v}\choose{\ell}}}^{2-1/\ell} \tag{since ${{d_v}\choose{\ell}} \cdot \ell^{\ell} \geq {d_v}^{\ell}$}\\
		&\leq 4m \cdot \ell^{2\ell} \cdot (\sS)^{2-1/\ell},
	\end{align*}
	where the last inequality is because $\sum_{v: d_v \geq \ell}  {{d_v}\choose{\ell}}$ is equal to $\sS$ when $\ell > 1$ and is equal to $2\cdot (\sS)$ when $\ell = 1$. 
\end{proof}

Using this lemma, the only change we need to do with our algorithm in Theorem~\ref{thm:alg-estimate} for improving its performance when counting stars is that instead of taking average of $O(\frac{m^{\ell}}{\sS})$ estimators, 
we only need to take average of $O(\frac{m}{\sS^{1/\ell}})$ many of them. The proof of correctness now follows exactly as before by Chebyshev's inequality (Proposition~\ref{prop:cheb}) as $\ell$ is a constant. With this minor modification, 
our algorithm then needs $O(\frac{m}{\sS^{1/\ell}})$ queries to compute a $(1\pm \eps)$-approximation of the number of occurrences of the star $S_\ell$ in any given graph $G$.

\end{document}